\documentclass[11pt]{paper}
\usepackage{fullpage}

\usepackage {amssymb}
\usepackage {amsmath}
\usepackage {amsthm}
\usepackage [usenames] {color}

\usepackage {enumerate}
\usepackage {plaatjes}
\usepackage {cite}
\usepackage{float}
\usepackage [left,pagewise] {lineno}
\usepackage {xspace}
\usepackage {wrapfig}

\definecolor {infocolor} {rgb} {0.6,0.6,0.6}
\linenumbers

\definecolor {sepia} {rgb} {0.75,0.30,0.15}
\everymath{\color{sepia}}

\floatstyle{ruled}
\newfloat{algorithm}{thp}{alg}
\floatname{algorithm}{Algorithm}

\newcommand {\mathset} [1] {\ensuremath {\mathbb {#1}}}
\newcommand {\R} {\mathset {R}}

\newcommand {\etal} {\textit {et al.}}
\newcommand {\eps} {\varepsilon}
\newcommand {\eqdef} {:=}
\newcommand {\boruvka}{Bor\r{u}vka}
\DeclareMathOperator {\wspd}{\texttt{wspd}}
\DeclareMathOperator {\as}{as}
\DeclareMathOperator {\emst}{emst}
\DeclareMathOperator {\argmin}{argmin}
\DeclareMathOperator {\DT}{DT}
\DeclareMathOperator {\UC}{UC}
\DeclareMathOperator {\LC}{LC}
\newcommand {\parent} {\overline}
\newcommand {\child} {\underline}

\newtheorem {theorem} {Theorem}[section]

\newtheorem {lem}[theorem] {Lemma}
\newtheorem {observation}[theorem] {Observation}
\newtheorem {cor}[theorem] {Corollary}
\newtheorem {claim}[theorem] {Claim}
\newtheorem {invariant}[theorem] {Invariant}

\title{\Large Triangulating the Square and Squaring the Triangle:\\ 
Quadtrees and Delaunay 
Triangulations are Equivalent\footnote{A preliminary version 
appeared in Proc.~22nd SODA, pp.~1759--1777, 2011}}

\author{Maarten L\"offler\thanks{%
   Department of Information and Computing Sciences,
   Universiteit Utrecht;
   3584 CC Utrecht,
   The Netherlands;
   \textsl{m.loffler@uu.nl}.
  }
\and
Wolfgang Mulzer\thanks{%
        Institut f{\"u}r Informatik,
  Freie Universit{\"a}t Berlin; 14195 Berlin, Germany;
  \textsl{mulzer@inf.fu-berlin.de}.
      }
}
\date{}
\begin{document}

\maketitle

\begin{abstract}
  We show that Delaunay triangulations and compressed quadtrees are equivalent
  structures. More precisely, we give two algorithms: the first computes
  a compressed quadtree for a planar point set, given the Delaunay
  triangulation; the second finds the Delaunay triangulation, given a
  compressed quadtree. Both algorithms run in deterministic linear time on
  a pointer machine.
  Our work builds on and extends previous
  results by Krznaric and Levcopolous~\cite{KrznaricLe98} and
  Buchin and Mulzer~\cite{BuchinMu11}. Our main tool for the second 
  algorithm is the
  well-separated pair decomposition (WSPD)~\cite{CallahanKo95},
  a structure that has been used previously to find Euclidean minimum
  spanning trees in higher dimensions~\cite{Eppstein00}. We show that knowing
  the WSPD (and a quadtree) suffices to compute
  a planar Euclidean minimum spanning tree (EMST) in \emph{linear} time. 
  With the EMST
  at hand, we can find the Delaunay triangulation in
  linear time~\cite{ChinWa99}.

  As a corollary, we obtain
  deterministic versions of many previous algorithms related
  to Delaunay triangulations, such as
  splitting planar Delaunay
  triangulations~\cite{ChazelleDeHuMoSaTe02,ChazelleMu11},
  preprocessing imprecise points for faster Delaunay
  computation~\cite{BuchinLoMoMuXX,LoefflerSn10}, and transdichotomous
  Delaunay triangulations~\cite{BuchinMu11,ChanPa09,ChanPa10}.
\end{abstract}

\section {Introduction}

  \tweeplaatjes [scale=0.95] {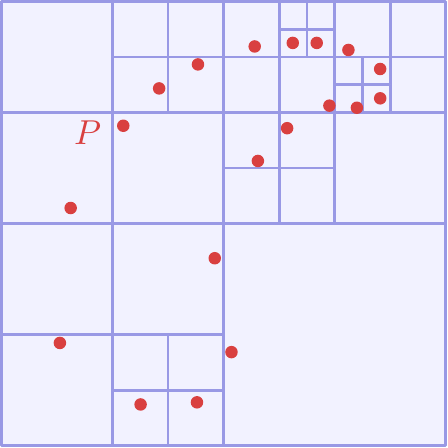} {intro-dt} {A planar point set $P$, and a quadtree (a) and a Delaunay triangulation (b) on it.}

  Delaunay triangulations and quadtrees are among the oldest and
  best-studied notions in computational
  geometry~\cite{deBergChvKrOv08,BoissonnatYv98,d-slsv-34,FinkelBe74,m-pdt-97,Samet90,ShamosHo75,PreparataSh85}, captivating the attention of researchers
  for almost four decades. 
  Both are proximity structures on planar point sets;
  Figure~\ref {fig:intro-qt+intro-dt} shows a simple example of these structures.
  Here, we will demonstrate that
  they are, in fact, equivalent in a very strong sense.
  Specifically, we describe two algorithms. The first computes a suitable 
  quadtree for $P$, given the Delaunay triangulation $\DT(P)$. 
  This algorithm closely follows
  a previous result by Krznaric and
  Levcopolous~\cite{KrznaricLe98}, who solve this problem in a stronger
  model of computation. Our contribution lies in adapting
  their algorithm to the real RAM/pointer machine model.\footnote{Refer
  to Appendix~\ref{app:models} for a description of different computational
  models.}
  The second algorithm, which is the main focus of this paper,
  goes in the other direction and computes
  $\DT(P)$, assuming that a suitable quadtree for $P$ is at hand.
  
  The connection between quadtrees and Delaunay triangulations was first
  discovered and fruitfully applied by Buchin and
  Mulzer~\cite{BuchinMu11} (see also~\cite{BuchinLoMoMuXX}).
  While their approach is to use a hierarchy of quadtrees
  for faster conflict location in a randomized incremental
  construction of $\DT(P)$,
  we pursue a strategy similar to the one by
  L{\"o}ffler and Snoeyink~\cite{LoefflerSn10}:
  we use the additional information
  to find a connected subgraph of $\DT(P)$,
  from which $\DT(P)$ can be computed in linear deterministic 
  time~\cite{ChinWa99}.
  As in L{\"o}ffler and Snoeyink~\cite{LoefflerSn10},
  our subgraph of choice is the
  \emph{Euclidean minimum spanning tree} (EMST) for $P$,
  $\emst(P)$~\cite{Eppstein00}.
  The connection between quadtrees and EMSTs is well known:
  initially, quadtrees were used to obtain fast approximations
  to $\emst(P)$ in high dimensions~\cite{CallahanKo93,Vaidya88}.
  Developing these ideas further, several algorithms  were found that
  use the \emph{well-separated pair decomposition}
  (WSPD)~\cite{CallahanKo95},
  or a variant thereof, to reduce EMST computation to
  solving the \emph{bichromatic closest pair} problem. In that
  problem, we are given
  two point sets $R$ and $B$, and we look for a
  pair $(r,b) \in R \times B$ that minimizes the distance
  $|rb|$~\cite{AgarwalEdScWe91,CallahanKo93,KrznaricLeNi99,Yao82}.
  Given a quadtree for $P$, a WSPD for $P$ can be found in
  linear time~\cite{BuchinLoMoMuXX,CallahanKo95,Chan08,HarPeled11}. 
  EMST algorithms based
  on bichromatic closest pairs constitute the fastest
  known solutions in higher dimensions. 
  Our approach is quite similar, but we focus exclusively
  on the plane. We use the quadtree and WSPDs to
  obtain a sequence of bichromatic closest pair problems,
  which then yield a sparse supergraph of the EMST. 
  There are several issues:  we need
  to ensure that the bichromatic closest pair problems
  have total linear size and can be solved in linear time,
  and we also need to extract the EMST from the supergraph
  in linear time.
  In this paper we show how to do this using
  the structure of the quadtree, combined with a partition of the
  point set according to angular segments similar to 
  Yao's technique~\cite{Yao82}.

  \subsection {Applications}

    Our two algorithms have several implications for derandomizing recent
    algorithms related to DTs. First,
    we mention \emph{hereditary} computation of DTs.
    Chazelle \etal~\cite {ChazelleDeHuMoSaTe02} show how to \emph{split} a Delaunay
    triangulation in linear expected time (see also~\cite{ChazelleMu11}).
    That is, given $\DT(P \cup Q)$,
    they describe a randomized algorithm to find $\DT(P)$ and
    $\DT(Q)$ in expected time $O(|P| + |Q|)$. Knowing that DTs
    and quadtrees are equivalent, this result becomes almost obvious, as quadtrees
    are easily split in linear time. More importantly, our
    new algorithm achieves linear \emph{worst-case} running time.
    Ailon \etal~\cite{AilonChClLiMuSe11} use hereditary
    DTs for 
    \emph{self-improving algorithms}~\cite{AilonChClLiMuSe11}.
    Together with the $\eps$-net construction by Pyrga and Ray~\cite{PyrgaRa08}
    (see~\cite[Appendix~A]{AilonChClLiMuSe11}),
    our result yields a deterministic version of their algorithm for
    point sets generated by a random source (the inputs are probabilistic,
    but not the algorithm).

    Eppstein \etal~\cite {EppsteinGoSu08} introduce the skip-quadtree and show
    how to turn a (compressed) quadtree into a skip-quadtree in linear
    time.  Buchin and Mulzer~\cite{BuchinMu11} use a
    (randomized) skip-quadtree to find the DT in linear
    expected time.
    This yields several improved results
    about computing DTs. Most notably, they show
    that in the \emph{transdichotomous} 
    setting~\cite{ChanPa09,ChanPa10,FredmanWi94}, computing DTs
    is no harder than sorting the points (according to some special order).
    Here, we show how to go directly from a quadtree
    to a DT, without skip-quadtrees or randomness.
    This gives the first \emph{deterministic} transdichotomous reduction from
    DTs to sorting.

    Buchin \etal~\cite {BuchinLoMoMuXX} use both hereditary
    DTs and the connection between skip-quadtrees and DTs
    to simplify and generalize
    an algorithm by L{\"o}ffler and Snoeyink~\cite{LoefflerSn10}
    to preprocess imprecise points for Delaunay triangulation
    in linear expected time (see also Devillers~\cite{Devillers11} for another
    simplified, but not worst-case optimal, solution). L{\"o}ffler and Snoeyink's
    original algorithm is deterministic, and the derandomized version of the
    Buchin~\etal~algorithm proceeds in a very similar spirit. However, we
    now have an optimal deterministic solution for the generalized 
    problem as well.

    \eenplaatje [scale=0.9] {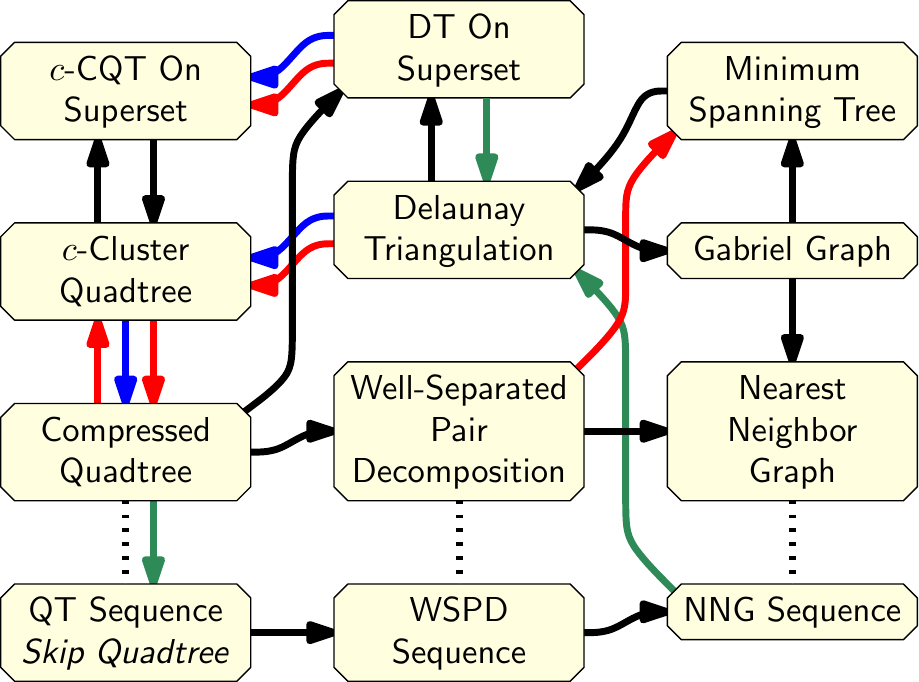}
    { We show which can be computed from which in linear time. The black
      arrows depict known linear time deterministic algorithms that work
      in the pointer machine/real RAM model. The red arrows depict our
      new results. Furthermore, for reference, we also show known 
      randomized linear time algorithms (in green) and known deterministic
      linear time algorithms that work in a weaker model of computation 
      (in blue).
    }

    In Figure~\ref {fig:diagram}, we show a graphical representation of 
    different proximity structures on planar point sets.
    The arrows show which structures can be computed from which in
    linear deterministic time on a pointer machine, before and after
    this paper.
    Please realize that there are several subtleties of different 
    algorithms and their interactions that are hard to show in a
    diagram, it is included purely as illustration of the impact of 
    our results.

  \subsection {Organization of this paper}
  
    The main result of our paper is an algorithm to compute a minimum spanning
    tree of a set of points from a given compressed quadtree. However, before
    we can describe this result in Section~\ref {sec:qt->dt}, we need to
    establish the necessary tools; to this end we review several known
    concepts in Section~\ref {sec:prelim} and prove some related technical
    lemmas in Section~\ref {sec:qt}.
    In Section~\ref {sec:dt->qt}, we describe the algorithm to compute a
    quadtree when given the Delaunay triangulation; this is an adaptation
    of the algorithm by Krznaric and Levcopoulos~\cite {KrznaricLe98} to the
    real RAM model.
    Finally, we detail some important implications of our two new algorithms
    in Section~\ref {sec:applications}.

\section {Preliminaries} \label {sec:prelim}

  We review some known definitions, structures, 
  algorithms, and their relationships.

  \subsection {Delaunay Triangulations and Euclidean Minimum Spanning Trees}

    Given a set $P$ of $n$ points in the plane, an important and extensively
    studied structure is the \emph {Delaunay triangulation} of $P$~\cite
    {deBergChvKrOv08,BoissonnatYv98,d-slsv-34,PreparataSh85,ShamosHo75}, 
    denoted $\DT(P)$.  It can be defined as the dual graph of the
    Voronoi diagram, the triangulation that optimizes the smallest angle in any
    triangle, or in many other equivalent ways, and it has been proven to
    optimize many other different criteria~\cite {m-pdt-97}.

    The \emph {Euclidean minimum spanning tree} of $P$, denoted $\emst(P)$, 
    is the tree of smallest total edge length that has the
    points of $P$ as its vertices, and it is well known that the EMST is a
    subgraph of the DT~\cite[Theorem~7]{ShamosHo75}.
    In the following, we will assume that all the pairwise distances in
    $P$ are distinct (a general position assumption), which implies that
    $\emst(P)$ is uniquely determined. Finally, we remind the reader that
    $\emst(P)$, like every minimum spanning tree, has 
    the following \emph{cut property}: let $P = R \cup B$
    a partition of $P$, and let $r$ and $b$ be the two points with $r \in R$ and
    $b \in B$ that minimize  the distance $|rb|$. Then $rb$ is an edge
    of $\emst(P)$. Note that this is very similar to the bichromatic closest 
    pair reduction mentioned in the introduction, but the cut property
    holds for any partition of $P$, whereas the bichromatic closest
    pair reduction requires a very specific decomposition of $P$ into pairs
    of subsets (which is usually not a partition).

 \subsection {Quadtrees---Compressed and $c$-Cluster}
  \label {sec:cqt&cqt}

    Let $P$ be a planar point set.
    The \emph{spread} of $P$ is defined as the
    ratio between the largest and the smallst distance between
    any two distinct points in $P$.
    A \emph{quad\-tree} for $P$ is a hi\-erarchical
    decomposition of an axis-aligned bounding square for $P$ into smaller
    axis-aligned
    \emph{squares}~\cite{deBergChvKrOv08,FinkelBe74,HarPeled11,Samet90}.
    A \emph{regular} quadtree is constructed by
    successively subdividing every square with at least two points
    into four congruent child squares. A node $v$ of a quadtree
    is associated with
    (i) $S_v$, the square corresponding to $v$;
    (ii) $P_v$, the points contained in $S_v$; 
    and (iii) $B_v$, the axis-aligned bounding square for
    $P_v$.
    $S_v$ and $B_v$ are stored explicitly at the node.
    We write $|S_v|$ and $|B_v|$  for the diameter of
    $S_v$ and $B_v$, and $c_v$ for the center of $S_v$.
    We will also use the shorthand $d(u,v) \eqdef d(S_u, S_v)$
    to denote the shortest distance between any point in $S_u$ and
    any point in $S_v$.
    Furthermore, we denote the parent of $v$ by $\parent v$.
    Regular quadtrees can have unbounded depth (if $P$ has unbounded 
    spread
    so in order to give any theoretical guarantees the concept is
    usually refined.
    In the sequel, we use two such variants of quadtrees, namely
    \emph{compressed} and \emph{$c$-cluster} quadtrees,
    which we show are in fact equivalent.

    A \emph {compressed} quadtree is a quadtree in which we replace
    long paths of nodes with only one child by a single 
    edge~\cite{BernEpGi94,BernEpTe99,BuchinLoMoMuXX,Clarkson83}.
    It has size $O (|P|)$.  Formally,
    given a large constant $a$, an $a$-compressed quadtree
    is a regular quadtree with additional \emph{compressed} nodes.\footnote
    {Such nodes are often called \emph {cluster}-nodes in the
    literature~\cite{BernEpGi94,BernEpTe99,BuchinLoMoMuXX},
    but we prefer the term \emph {compressed} to avoid confusion with
    $c$-cluster quadtrees defined below.}
    A compressed
    node $v$ has only one child $\child v$ with
    $|S_{\child v}| \leq |S_v|/a$ and such that $S_v \setminus S_{\child v}$
    has no points from $P$. 
    Figure~\ref {fig:ex-compquad-notaligned} shows an example.
     Note that in our definition
    $S_{\child v}$ need not be aligned with $S_v$, which would
    happen if we literally ``compressed'' a regular quadtree.
    This relaxed definition is necessary because existing algorithms for 
    computing
    aligned compressed quadtrees use a more powerful model of
    computation than our real RAM/pointer machine
    (see Appendix~\ref{app:models}).  In the usual applications of quadtrees,
    this is acceptable. In fact, Har-Peled~\cite[Chapter~2]{HarPeled11}
    pointed out that some non-standard operation is
    \emph{inevitable} if we require that the squares of the compressed quadtree
    are perfectly aligned. However, here we intend to derandomize algorithms
    that work on a traditional real RAM/pointer machine, so we
    prefer to stay in this model. This keeps our results comparable with the
    previous work.

    \drieplaatjes {ex-compquad-notaligned} {ex-ccluster} {ex-cclusquad} 
    { (a) A compressed quadtree on a set of $15$ points.
      (b) A $c$-cluster tree on the same point set.
      (c) In a $c$-cluster quadtree, the internal nodes of the $c$-cluster tree are replaced by quadtrees.      
    }    

    Now let $c$ be a large enough constant.
    A subset $U \subseteq P$ is a \emph{$c$-cluster} if $U = P$ or
    $d(U, P \setminus U) \geq c |B_U|$, where $B_U$ denotes the
    smallest axis-aligned bounding square
    for $U$, and $d(A,B)$ is the minimum distance between a point in $A$
    and a point in $B$~\cite{KrznaricLe95,KrznaricLe98}. In other words, 
    $U$ is a $c$-cluster precisely if $\{U, P \setminus U\}$ is a
    $(1/c)$-\emph{semi}-separated pair~\cite{HarPeled11,Varadarajan98}.
    It is easily seen that the $c$-clusters for $P$ form a laminar
    family, i.e., a set system in which any two sets $A$ and $B$  satisfy either
    $A \cap B = \emptyset$; $A \subseteq B$; or $B \subseteq A$.
    Thus, the $c$-clusters define a \emph{$c$-cluster tree} $T_c$.
    Figure~\ref {fig:ex-ccluster} shows an example.
    These trees are a very natural way to tackle
    point sets of unbounded spread, and they have linear size. However,
    they also may have high degree. To avoid this, a $c$-cluster tree
    $T_c$ can be augmented by additional nodes, adding more structure
    to the parts of the point set that are not strongly clustered.
    This is done as follows. First, recall that a 
      quadtree is called \emph {balanced} if for every node 
      $u$ that is either a leaf or a compressed node,
     the square $S_u$ is adjacent only to squares that are within a factor $2$
     of the size of $S_u$.\footnote{We remind the reader that in our 
     terminology,
     a \emph{compressed node} is the node whose square contains a much
     smaller quadtree, and not the root node of the smaller quadtree.}
    For each
    internal node $u$ of $T_c$ with set of children $V$, we
    build a balanced regular quadtree on a set of points
    containing one representative point from each node in $V$
    (the intuition being that such a cluster is so small and far from its
    neighbors, that we might as well treat it as a point). This
    quadtree has size $O(|V|)$ (Lemma~\ref{lem:c-cluster-QT}), so
    we obtain a tree of constant degree and linear size,
    the \emph {$c$-cluster quadtree}.
    Figure~\ref {fig:ex-cclusquad} shows an example.
    The sets $P_v$, $S_v$ and $B_v$
    for the $c$-cluster quadtree are just as for regular and compressed
    quadtrees, where in $P_v$ we expand the representative points
    appropriately.
    Note that it is possible that $S_v \nsupseteq P_v$, but the points 
    of $P_v$ can never be too far from $S_v$.
    In Section~\ref {sec:ccqt} we elaborate more on $c$-cluster quadtrees 
    and their properties, and in Section~\ref{sec:compressed-c-cluster}, 
    we prove that $c$-cluster quadtrees and compressed quadtrees
    are equivalent (Theorem~\ref{thm:cluster-compressed-equiv}). 

 \subsection {Well-Separated Pair Decompositions}

    For any two finite sets $U$ and $V$, let
    $U \otimes V \eqdef \{\{u,v\} \mid u \in U, v \in V, u \neq v\}$.
    A \emph {pair decomposition}
    $\mathcal{P}$ for a planar\footnote
    {Although some of these notions extend naturally to higher dimensions,
    the focus of this paper is on the plane.}
    $n$-point set $P$ is a set of $m$ \emph {pairs}
    $\{\{U_1, V_1\},$ $\ldots,$ $\{U_m,V_m\}\}$, such
    that
    (i) for all $i = 1, \ldots, m$, we have $U_i, V_i \subseteq P$ and
       $U_i \cap V_i = \emptyset$; and
    (ii) for any $\{p, q\} \in P \otimes P$, there is exactly one $i$ with
      $\{p, q\} \in U_i \otimes V_i$.
    We call $m$ the \emph{size} of $\mathcal{P}$.
    Fix a constant $\eps \in (0,1)$, and let $\{U, V\} \in \mathcal{P}$.
    Denote by $B_U$, $B_V$ the smallest axis-aligned squares containing
    $U$ and $V$. We say that $\{U,V\}$
    is \emph{$\eps$-well-separated} if $\max\{|B_U|, |B_V|\} \leq \eps 
    d(B_U,B_V)$,
    where $d(B_U,B_V)$ is the distance between $B_U$ and $B_V$ (i.e.,
    the smallest distance
    between a point in $B_U$ and a point in $B_V$).
    If $\{U,V\}$ is not $\eps$-well-separated, we say it is
    \emph{$\eps$-ill-separated}.
    We call $\mathcal{P}$ an
    \emph {$\eps$-well-separated pair decomposition} ($\eps$-WSPD)
    if all its pairs are
    $\eps$-well-separated~\cite{CallahanKo93,CallahanKo95,Eppstein00,HarPeled11}.

    Now let  $T$ be a (compressed or $c$-cluster) quadtree for $P$.
    Given $\eps > 0$, it is well known that $T$ can be used to
    obtain an $\eps$-WSPD for $P$ in linear time~\cite{CallahanKo95,HarPeled11}.
    Since we will need some specific properties of such an $\eps$-WSPD,
    we give pseudo-code for such an  algorithm 
    in Algorithm~\ref{alg:wspd}. We call this algorithm
    $\wspd$, and denote its output on input $T$ by $\wspd(T)$.
    The correctness of the algorithm $\wspd$ is immediate,
    since it only outputs well-separated pairs, and the bounds on the
    running time and the size of $\wspd(T)$ follow from a well-known
    volume argument which we
    omit~\cite{BuchinLoMoMuXX,CallahanKo95,Chan08,HarPeled11}.
   \begin{algorithm}[ht]
    \begin{enumerate}
    \item Call $\wspd(r)$ on the root $r$ of $T$.
    \end{enumerate}
    $\wspd(v)$
    \begin{enumerate}
    \item If $v$ is a leaf, return $\emptyset$.
    \item Return the union of $\wspd(w)$ and
      $\wspd(\{w_1, w_2\})$ for all children $w$ and
      pairs of distinct children $w_1, w_2$ of $v$.
    \end{enumerate}
    \vskip0.2cm
    \noindent
    $\wspd(\{u, v\})$
    \begin{enumerate}
    \item If $S_{u}$ and $S_{v}$ are
    $\eps$-well-separated, return $\{u, v\}$.
    \item Otherwise, if $|S_u|\leq |S_{v}|$, return
       the union of $\wspd(\{u, w\})$ for all children $w$ of $v$.
    \item  Otherwise, return the union of $\wspd(\{w, v\})$ for all
       children $w$ of $u$.
    \end{enumerate}
    \caption{Finding a well-separated pair decomposition.}
    \label{alg:wspd}
    \end{algorithm}

    \begin{theorem}\label{thm:wspd}
       There is an algorithm \emph{$\wspd$}, that given a 
       (compressed or $c$-cluster) 
       quadtree $T$ for a planar $n$-point set $P$, finds in time $O(n)$
       a linear-size $\eps$-WSPD for $P$, denoted \emph{$\wspd(T)$}.
       \qed
    \end{theorem}

    Note that the WSPD is not stored explicitly: we cannot afford to
    store all the pairs $\{U,V\}$, since their total size might be
    quadratic. Instead, $\wspd(T)$ contains pairs $\{u,v\}$, where $u$
    and $v$ are nodes in $T$, and $\{u,v\}$  is used to represent
    the pair $\{P_u, P_v\}$.

    Note that the algorithm computes the WSPD with respect to the squares
    $S_v$, instead of the bounding squares $B_v$. This makes no big difference,
    since for compressed quadtrees $B_v \subseteq S_v$, and for
    $c$-cluster quadtrees $B_v$ can be outside $S_v$ only for
    $c$-cluster nodes, resulting in a loss of at most a factor $1+1/c$
    in separation. 
    Referring to the pseudo-code in Algorithm~\ref{alg:wspd}, 
    we now prove three observations. The first observation says
    that the size of the squares under consideration strictly 
    decreases throughout the algorithm.

    \begin{observation} \label{obs:parents-bigger}
    Let $\{u,v\}$ be a pair of distinct nodes of $T$. If 
    \emph{$\wspd(\{u, v\})$}
    is executed by \emph{$\wspd$} run
    on $T$ (in particular, if \emph{$\{u,v\} \in \wspd(T)$}),
    then $\max\{|S_u|, |S_v| \} \leq \min \{|S_{\parent u}|, |S_{\parent v}|\}$.
    \end{observation}

    \begin{proof}
    We use induction on the depth of the call stack for
    $\wspd(\{u,v\})$. Initially, $u$ and $v$ are children 
    of the same node, and the statement holds.
    Furthermore, assuming that $\wspd(\{u,v\})$ is called by
    $\wspd(\{u, \parent v\})$ (and hence $|S_u| \leq |S_{\parent v}|$), we get
    $\max\{|S_u|,|S_v|\} \leq |S_{\parent v}| =
    \min\{|S_{\parent u}|, |S_{\parent v}|\}$,
    where the last equation follows by induction.
    \end{proof}

    The next observation states that the wspd-pairs reported by the
    algorithm are, in a sense, as high in the tree as possible.
    \begin{observation}\label{obs:parents-not-ws}
    If \emph{$\{u, v\} \in \wspd(T)$},
    then $\parent u$ and $\parent v$ are ill-separated.
    \end{observation}

    \begin{proof}
    If $\parent u = \parent v$, the claim is obvious. Otherwise, let us assume
    that $\wspd(\{u,v\})$ was  called by $\wspd(\{u, \parent v\})$. This means
    that $\{u, \parent v\}$ is ill-separated and
    $\max\{|S_u|, |S_{\parent v}|\} = |S_{\parent v}|$.
    Therefore, 
    $ 
    \max\{|S_{\parent u}|, |S_{\parent v}|\} \geq |S_{\parent v}| >
    \eps d(u,\parent v) \geq \eps d(\parent u, \parent v) 
    $,  
    and $\{\parent u, \parent v\}$ is ill-separated.
    \end{proof}

   The last claim shows that for each wspd-pair, we can find
   well-behaved boxes whose size is comparable to the distance
   between the point sets. In the following, this will be a useful tool
   for making volume arguments that bound the number of wspd-pairs to consider.
\begin{claim}\label{clm:quad-squares}
Let \emph{$\{u,v\} \in \wspd(T)$}.
Then there exist squares $R_u$ and $R_v$ such that
(i) $S_u \subseteq R_u \subseteq S_{\overline{u}}$ and
    $S_v \subseteq R_v \subseteq S_{\overline{v}}$;
(ii) $|R_u| =  |R_v|$;
and (iii) $|R_u|/2\eps \leq d(R_u,R_v) \leq 2|R_u|/\eps$.
\end{claim}

\begin{proof}
Suppose $\wspd(\{u,v\})$ is called by
$\wspd(\{u, \overline{v}\})$, the other case is symmetric.
Let us define $r \eqdef \min\{\eps d(u,v), |S_{\overline{v}}|\}$.
By Observation~\ref{obs:parents-bigger},
we have $|S_u|,|S_v| \leq |S_{\overline{v}}| \leq |S_{\overline{u}}|$.
 Since $\{u,v\}$
is well-separated, we have $\eps d(u,v) \geq \max\{|S_u|, |S_v|\}$.
Hence, $|S_{\overline{u}}|, |S_{\overline{v}}| \geq r \geq |S_u|, |S_v|$,
and we can pick squares $R_u$ and $R_v$ of diameter $r$ that fulfill (i).
Now (ii) holds by construction, and it remains to check (iii).
First, note that $d(R_u, R_v) \geq d(u,v) - 2r
\geq (1-2\eps) d(u,v) \geq r/2\eps$, for $\eps \leq 1/4$. This
proves the lower bound.
For the upper bound, observe that
$\eps d(u,v) \leq \eps(d(u,\overline{v})+|S_{\overline{v}}|)
\leq (1+\eps)|S_{\overline{v}}|$, because $\{u,\overline{v}\}$
is ill-separated. Thus, we have $\eps d(u,v)/2 \leq r$,
and $d(R_u, R_v) \leq d(u,v) \leq 2r/\eps$, as desired.
\end{proof}

\section{More on Quadtrees} \label {sec:qt}

  In this section, we describe a few more properties of the $c$-cluster
  trees and $c$-cluster quadtrees defined in Section~\ref{sec:cqt&cqt}, 
  and we prove that they are equivalent to the more standard compressed
  quadtrees (Theorem~\ref {thm:cluster-compressed-equiv}).
  Since most of the material is very technical, we encourage 
  the impatient reader to skip ahead to Section~\ref {sec:qt->dt}.

  \subsection {$c$-Cluster Quadtrees} \label {sec:ccqt}

    Krznaric and Levcopolous~\cite[Theorem~7]{KrznaricLe95} showed that a
    $c$-cluster tree can be computed in linear time from a Delaunay 
    triangulation.
    \begin{theorem}[Krznaric-Levcopolous]\label{thm:c-cluster-tree}
      Let $P$ be a planar $n$-point set.
      Given a constant $c \geq 1$ and $\DT(P)$, we can find a $c$-cluster
      tree $T_c$ for $P$ in $O(n)$ time and space on a pointer machine.
      \qed
    \end{theorem}

    Here, we will actually use a more relaxed notion of $c$-cluster
    trees: let $c_1$, $c_2$ be two constants with
    $1 \leq c_1 \leq c_2$, and let $P$ be a planar $n$-point set.
    A \emph{$(c_1, c_2)$-cluster tree} $T_{(c_1, c_2)}$ is a rooted
    tree in which each inner node has at least two children and
    which has $n$ leaves, one for each point in $P$. 
    Each node $v \in T_{(c_1,c_2)}$ corresponds to a subset $P_v \subseteq P$
    in the natural way. Every node $v$ must fulfill
    two properties:
    (i) if $v$ is not the root, then
      $d(P_v, P \setminus P_v) \geq c_1|B_{P_v}|$; and
    (ii) if $P_v$ has a proper subset $Q \subset P_v$ with
       $d(Q, P \setminus Q) \geq c_2|B_Q|$, then
       there is a child $w$ of $v$ with $Q \subseteq P_w$.
    In other words, each node of $T_{(c_1, c_2)}$ corresponds to
    a $c_1$-cluster of $P$, and $T_{(c_1, c_2)}$ must have a node for every
    $c_2$-cluster of $P$.  Thus, the original $c$-cluster tree is also 
    a $(c,c)$-cluster tree. Our relaxed definition allows for some flexibility
    in the construction of $T_{(c_1, c_2)}$ while providing the same
    benefits as the original  $c$-cluster tree. Thus, outside this section
    we will be slightly sloppy and not distinguish between $c$-cluster trees
    and $(c, \Theta(c))$-cluster trees.

      As mentioned above, the tree $T_{(c_1, c_2)}$ is quite similar 
      to a well-separated pair decomposition: any two unrelated nodes
      in $T_{(c_1,c_2)}$ correspond to a $(1/c_1)$-well-separated pair. 
      However, $T_{(c_1, c_2)}$ has the huge drawback that it may contain
      nodes of unbounded degree.  For example, if
      the points in $P$ are arranged in a square grid,
      then $T_{(c_1,c_2)}$ consists of a single root with $n$ children.
      Nonetheless, $T_{(c_1,c_2)}$ is still useful, since it represents
      a decomposition of $P$ into well-behaved pieces. 
      As explained above,  the $(c_1, c_2)$-cluster quadtree
      $T$ is obtained by augmenting $T_{(c_1, c_2)}$
      with quadtree-like pieces to replace the nodes with many children.

      We will now prove some relevant properties of $(c_1,c_2)$-cluster
      quadtrees. 
      For a node $u$ of $T_{(c_1,c_2)}$, let $T^Q_u$ be the balanced regular 
      quadtree on the representative points of $u$'s children.
      The \emph{direct neighbors} of a square $S$ in $T_u^Q$ 
      are the $8$ squares of size $|S|$ that surround $S$.
      First, we recall how the balanced tree $T^Q_u$ is obtained: we start with
      a regular (uncompressed) quadtree $T'$ for the representative points.
      While $T'$ is not balanced, we take a leaf square $S$ of $T'$ that
      is adjacent to a leaf square of size less than $|S|/2$ and
      we split $S$ into four congruent child squares. The following
      theorem is well known.
     \begin {theorem} [Theorem~14.4 of \cite{deBergChvKrOv08}]
    \label {thm:balance}
      Let $T'$ be a quadtree with $m$ nodes. 
      The above procedure yields a balanced
      quadtree with $O(m)$ nodes, and it can be implemented
      to run in $O(m)$ time.
      \qed
    \end {theorem}

    Let $v$ be a child of $u$ in $T_{(c_1,c_2)}$.
    The properties of the balanced quadtree $T_u^Q$ and the
    fact that the children of $u$ are mutually well-separated yield 
    the following observation.
    
      \begin {observation} \label {obs:close}
        If $c_1$ is large enough,
        at most four leaf squares of $T^Q_u$ contain points from $P_v$.
      \end {observation}
      
      \begin {proof}
        Let $d \eqdef |B_v|$ be the diameter of the bounding square for 
	$P_v$. By definition, $P_v$ is a $c_1$-cluster,
        so the distance from any point in $P_v$ to any point
        in $P \setminus P_v$ is at least $c_1d$.
	Suppose that $S$ is a leaf square of $T^Q_u$ with
	$S \cap P_v \neq \emptyset$, and let $\overline{S}$ be the parent of $S$.

	There are two possible reasons for the creation of $S$:
	either $S$ is part of the original regular quadtree for
	the representative points, or $S$ is generated during the
	balancing procedure.
	 In the former case, 
	$\overline{S}$ contains at least two representative points.  
	Thus, since in $\overline{S}$ there is a point from $P_v$ and 
	a point from $P \setminus P_v$,
	we have $|S| \geq c_1d/2$. In the latter case, 
	$\overline{S}$ must be a direct neighbor of a square with
	at least two representative points 
	(see~\cite[Proof of Theorem~14.4]{deBergChvKrOv08}).
	Therefore, since $\overline{S}$  contains
	a point from $P_v$ and has a direct neighbor with a point
	from $P \setminus P_v$, the diameter of $S$ is at least
	$c_1d/4$. Either way, we certainly have 
	$|S| \geq c_1d/4$.

        Now if $c_1 \geq 8$, then $c_1d/4 \geq 2d$, so the side length
	of every leaf square $S$ that intersects $P_v$ is strictly larger than
	$d$. Thus, $P_v$ can be covered by at most $4$ such squares,
	and the claim follows. 
      \end{proof}

      To see that $(c_1, c_2)$-cluster quadtrees have linear size, we need 
      a property that is (somewhat implicitly)
      shown in~\cite[Section~4.3]{KrznaricLe98}.

      \begin{lem}\label{lem:c-cluster-QT}
      If $u$ has $m$ children $v_1$, $v_2$, $\ldots$, $v_m$ in $T_c$, 
      then $T^Q_u$ has $O(m)$ nodes.
      \end{lem}

      \begin{proof}
      Note that the total number of nodes in $T_u^Q$ is
      proportional to the number of 
      squares that contain at least two representative points.
      Indeed, the number of squares in a balanced regular quadtree
      is proportional to the number of squares in the corresponding 
      unbalanced regular quadtree (Theorem~\ref{thm:balance}), and in that tree 
      the squares with
      at least two points correspond to the internal nodes, each
      of which has exactly four children.
      Thus, it suffices to 
      show that the number of squares in $T_u^Q$ with at least two representative 
      points is $O(m)$.

      Call a square $S$ of $T_u^Q$ \emph{full} if $S$ contains a representative
      point. 
      A full square $S \in T_u^Q$ is called \emph{merged} if it has
      at least two full children.  There are
      $O(m)$ merged squares, so 
      we only need to bound the number of non-merged full squares
      with at least two points. 
      These squares can be charged to the merged squares,
      using the following claim.
     
      \begin{claim}\label{clm:merged_ancestor}
       There exists a constant $\beta$ (depending on $c_2$) such
       that the following holds:
      for any full square $S$ with at least two representative points,
      one of the $\beta$ closest ancestors of $S$
      in $T_u^Q$ (possibly $S$ itself) is either merged or 
      has a merged direct neighbor.
      \end{claim}
      \begin{proof}
      Let $S$ be a non-merged full square with at least two representative points.
      Since $S$ intersects more than one $P_{v_i}$, the definition of
      $T_{(c_1,c_2)}$ implies that the set 
      $S \cap P_u$ is not a $c_2$-cluster. Thus, 
      $P_u \setminus S$ contains a point at distance at most $c_2|S|$ from $S$.
      Hence, $S$ has an ancestor $S'$ in $T_u^Q$ that is at most
      $O(\log c_2)$ levels above $S$ and that has a full direct 
      neighbor $S'' \not= S'$ (note that $T_u^Q$ is balanced, so 
      $S''$ actually belongs to $T_u^Q$). 

      We repeat the argument: since
      $(S' \cup S'') \cap P_u$ is not a $c_2$-cluster, there is a point in
      $P_u \setminus (S' \cup S'')$ at distance
      at most $c_2|S' \cup S''| \leq 2c_2|S'|$ from $S' \cup S''$. Thus, if we go up
      $O(\log c_2)$ levels in $T_u^Q$, we either encounter a common ancestor
      of $S'$ and $S''$, in which case we are done, or we have found
      a set $\mathcal{S}$ of
      three full squares of $T_u^Q$ such that (i) one square in
      $\mathcal{S}$ is an ancestor of $S$; (ii) the squares in $\mathcal{S}$
      have equal size;  and (iii) the squares in $\mathcal{S}$ 
      form a (topologically) connected set.

      We keep repeating the argument while going up the tree. 
      In each step, if we do not encounter a common ancestor of
      at least two squares in $\mathcal{S}$,
      we can add one more full square to $\mathcal{S}$.
      However, as soon as we have five squares of equal size
      that form a connected set, at least two of them have a common
      parent. Thus, the process stops after at most two more
      iterations. Furthermore, since $\mathcal{S}$ is connected, once at least two 
      squares in $\mathcal{S}$ have a common parent, the parents of the other
      squares must be direct neighbors of that parent. Hence, 
      we found an ancestor of $S$ that is
      only a constant number of levels above $S$ and that is merged
      or has a merged direct neighbor, as desired.
      \end{proof}
      
      Now we use Claim~\ref{clm:merged_ancestor} to charge
      each non-merged full node with at least two representative points
      to a merged node.
      Each merged node is charged at most 
      $9 \cdot 4^\beta = O(1)$ times, and Lemma~\ref{lem:c-cluster-QT} 
      follows.
      \end{proof}

  The proof of Lemma~\ref{lem:c-cluster-QT} implies the following, slightly
  stronger claim: Recall that $T^Q_u$ was constructed by building
  a regular quadtree for the representative points for $u$'s children,
  followed by a balancing step. Now, suppose that before the balancing
  step we subdivide each leaf that contains a representative
  point for a $c$-cluster $C$ until it has size at 
  most $\alpha d(C, P \setminus C)$, for some constant $\alpha > 0$ (if the 
  leaf is smaller than $\alpha d(C, P \setminus C)$, we do nothing).
  Call the tree that results after the balancing step $T_2$.
  \begin{cor}\label{cor:c-cluster-QT-extended}
      The tree $T_2$ has $O(m)$ nodes.
  \end{cor}

  \begin{proof}
    We only need to worry about the additional squares created during the
    subdivision of the leaves. If we take such a square and go up
    at most $\log(1/\alpha)$ levels in the tree, we get a square with a direct
    neighbor that contains a point from another cluster. Now the 
    argument from the proof of Lemma~\ref{lem:c-cluster-QT} applies
    and we can charge the additional squares to merged squares, as before.
  \end{proof}
  \subsection{Balancing and Shifting Compressed Quadtrees}
    \label{sec:shiftbalance}

\tweeplaatjes {shift-plain} {shift-comp}
{
  (a) A regular quadtree on a set of $8$ points.
  (b) A slight shift of the base square may cause many new compressed nodes in the quadtree.
}

    In this section, we show that it is possible to ``shift'' a quadtree; that is, given a compressed quadtree on a set of points $P$ with base square $R$, to compute another compressed quadtree on $P$ with a base square that is similar to $R$, in linear time.
    The main difficulty lies in the fact that the clusters in the two quadtrees can be very different, as illustrated in Figure~\ref {fig:shift-plain+shift-comp}.

    \begin {theorem} \label {thm:shiftbalance}
      Suppose $a$ is a sufficiently large constant and
      $P$ a planar $n$-point set.
      Furthermore, let $T$ be an $a$-compressed quadtree
      for $P$ with base square $R$, and let $S$ be a square with
      $S \supseteq P$ and $|S| = \Theta(|R|)$. 
      Then we can construct in $O(m)$ time a balanced $a$-compressed quadtree $T'$
      for $P$ with base square $S$ and with $O(m)$ nodes.
    \end {theorem}

    The idea is to construct $T'$ in the traditional way through
    repeated subdivision of the base square $S$, while using the information
    provided by $T$ in order to speed up the point location. 
    We will use the terms \emph{$T$-square} and \emph{$T'$-square} to
    distinguish the squares in the two trees.
    During the subdivision process, we maintain the partial tree $T'$, 
    and for each square $S'$ of $T'$ 
    we keep track of the $T$-squares that have similar size
    as $S'$ and that intersect $S'$ (in an associated set). 
    We call the leaves of the current partial tree the
    \emph{frontier} of $T'$. In each step, we
    pick a frontier $T'$-square and split it, until we have reached
    a valid quadtree for $P$. 
    We need to be careful in order to keep $T'$ 
    balanced and in order to deal with compressed nodes. The former problem
    is handled by starting a cascading split operation as soon as 
    a single split makes $T'$ unbalanced.
    For the latter problem, we would like to treat the compressed children
      in the same way as the points in $P$, and handle them
      later recursively. However, there is a problem: during the
      balancing procedure, it may happen that a compressed child becomes
      too large for its parent square and should be part of the regular
      tree. In order to deal with this, we must keep track of the compressed
      children in the associated sets of the $T'$-squares. When we detect
      that a compressed child has become too large for its parent, we treat
      it like a regular square. Once we are done, we recurse on the 
      remaining compressed children. Through a charging scheme, we can show
      that the overall work is linear in the size of $T$.
      The following paragraphs describe the individual steps of the algorithm 
      in more detail.

   \paragraph{Initialization and Data Structures.}
      We obtain from  $S$ a grid with squares of size in
       $(|R|/2, |R|]$, either by repeatedly subdividing 
      $S$, if $|S| > |R|$; or by repeatedly doubling $S$,
      if $|S| \leq |R|/2$. Since $|S| = \Theta(|R|)$, this requires a constant
      number of steps. Then we determine
      the $T'$-squares $S_1', \dots, S_k'$ of that grid that intersect $R$ 
      (note that $k \leq 9$). Our algorithm maintains the 
      following data structures:
      (i) a list $L$ of \emph{active} $T'$-squares; and (ii) for each 
      $T'$-square $S'$ a list $\as(S')$ of \emph{associated} $T$-squares. 
      We will maintain the invariant that $\as(S')$ contains the smallest 
      $T$-squares that have size at least $|S'|$  and that intersect $S'$, as
      well as any compressed children that are contained in such a
      $T$-square and that intersect $S'$. This invariant implies
      that each $S'$ has $O(1)$ associated squares. We call
      a $T'$-square $S'$ \emph{active} if 
      $\as(S')$ contains a $T$-square
      of size in $[|S'|, 2|S'|)$ or a compressed child of size in 
      $[|S'|/2^{2a}, |S'|)$. 
      Initially, we set $L \eqdef \{S_1', \dots, S_k'\}$ 
      and $\as(S_1') = \as(S_2') = \dots = \as(S_k') = \{R\}$,
      fulfilling the invariant.
  
    \paragraph{The Split Operation.}
      The basic operation of our algorithm is the \emph{split}.
      A split takes a $T'$-square $S'$ and subdivides it into
      four children $S'_1, \dots, S'_4$. Then it computes
      the associated sets $\as(S'_1), \dots, \as(S'_4)$ as follows.
      For each $i=1,\dots,4$, we intersect 
      $S'_i$ with all $T$-squares in $\as(S')$, and we put those $T$-squares 
      into $\as(S_i')$ that have non-empty intersection with $S'_i$.
      Then we replace each $T$-square in $\as(S')$ that is neither a 
      leaf, nor a compressed node, nor a compressed child by those of its
      children that have non-empty intersection with $S_i'$. Finally, we remove
      from $\as(S_i')$ those compressed nodes whose compressed children have
      size at least $|S_i'|$ and intersect $S_i'$. Having determined
      $\as(S'_i)$, we use it to check whether $S_i'$ is active. If so, we add 
      it to $L$. The split operation maintains the invariant
      about the associated sets, and it takes constant time.
      
    \paragraph{Main Body and Point-Location.}
      We now describe the main body of our algorithm. It consists of 
      \emph{phases}. In each phase, we remove a $T'$-square 
      $S'$ from $L$. We perform a split operation on $S'$ as
      described above. Then, we start the \emph{balancing procedure}. 
      For this, we check the four $T'$-squares in the
      current frontier that are directly above, below, to the left 
      and to the right of $S'$ to see whether any of 
      them have size $2|S'|$.
      We put each such $T'$-square into a queue $Q$. Then, while $Q$
      is not empty, we remove a square $N'$ from $Q$ and perform a
      split operation on it (note that this may create new active
      squares). Furthermore, if $N'$ is in $L$, 
      we remove it from $L$.  Finally, we consider the
      $T'$-squares of the current frontier directly above, below, to 
      the left and to the right
      of $N'$. If any of them have size $2|N'|$ and are not in $Q$ yet, we 
      append them to $Q$ and continue. The balancing procedure, and 
      hence the phase, ends once $Q$ is empty.

      We continue this process until $L$ is empty. Next, we do
       \emph{point-location}. Let $S'$ be a
      $T'$-square of the current frontier. Since $L$ is empty, 
      $S'$ is associated with $O(1)$ $T$-squares, all of 
      which are either leaves or compressed nodes or 
      compressed children in $T$. For each $T$-leaf that intersects 
      $S'$, we determine whether it contains a point that lies 
      in $S'$. In the end, we have a set of
      at most four points from $P$  or compressed children of $T$
      that intersect $S'$, and we call this set the \emph{secondary} associated
      set for $S'$, denoted by $\as_2(S')$. We do this for every $T'$-square
      in the
      current frontier. 
      
      \paragraph{The Secondary Stage.}
      Next, the goal is to build a small compressed quadtree for the secondary
      associated set of each square in the current frontier. Of course,
      the tree needs to remain balanced.
      For this, we start an operation that is similar to the main body of the 
      algorithm.  We call a $T'$-square $S'$ \emph{post-active} if 
      $|\as_2(S')| \geq 2$ and the smallest bounding
      square for the elements in $\as_2(S')$ has size larger than $|S'|/128a$.
      We put all the post-active squares into a list $L_2$ and we proceed
      as before: we repeatedly take a post-active square from $L_2$,
      split it, and then perform a balancing procedure. Here,
      the splitting operation is as follows: given
      a square $S'$, we split it into four children $S'_1, \ldots, S'_4$.
      By comparing each child $S_i'$ to each element in the secondary associated
      set $\as_2(S')$, we determine the new secondary associated sets 
      $\as_2(S'_1), \ldots, \as_2(S'_4)$. We use these associated sets to check
      which children $S'_i$ (if any) are post-active and add them to $L_2$,
      if necessary. This splitting
      operation takes constant time. Again, it may
      happen that the balancing procedure creates new post-active
      squares. We repeat this procedure until 
      $L_2$ is empty.

      \paragraph{Setting Up the Recursive Calls.}
      After the secondary stage, there are no more post-active squares, so 
      for each square $S'$ in the current frontier
      we have (i) $|\as_2(S')| \leq  1$; or (ii) the smallest bounding
      square of $\as_2(S')$ has size at most $|S'|/128a$.
      Below in Lemma~\ref{lem:secondary-bound} 
      we will argue that if $\as_2(S')$ contains a single
      compressed child $C$, then $C$ has size at most
      $|S'|/128a$. Thus, (ii) holds in any case. 
      The goal now is to set up a recursive call of the algorithm to handle
      the remaining compressed children. Unfortunately, a compressed
      child may intersect several leaf $T'$-squares,
      so we need to be careful about choosing the base squares for the
      recursion. 

      \tweeplaatjes {shift-rec} {shift-rec2}
      { (a) A frontier square $S'$ of $T'$ intersects several compressed children of $T$. We identify the list $X$ of $T'$ squares that intersect the same children.
        (b) To apply the shifting algorithm recursively, we choose base squares $\tilde R$ and $\tilde S$ aligned with $T$ and $T'$.
      }
      
      Let $S'$ be a square of the current frontier, and set $X \eqdef \{S'\}$.
      While there is a compressed child $C$ in 
      $\as_2(X) \eqdef \bigcup_{S'' \in X} \as_2(S'')$ that intersects 
      the boundary of 
      $S(X) \eqdef \bigcup_{S'' \in X} S''$, we add all the $T'$-squares 
      of the current frontier
      that are intersected by $C$ to $X$.  Since $T'$ is balanced, the 
      $i$-th square $S^{(i)}$ that we add to $X$ has size at most $2^i|S'|$ and
      hence the bounding square of $\as_2(S^{(i)})$  has size at most 
      $2^i|S'|/128a$. By construction, $\as_2(S^{(i)})$ contains at least
      one element that intersects a square in the old $X$, so by induction
      we know that after $i$ steps the set
      $\as_2(X)$ has a bounding square of size at most
      $2^{i+1}|S'|/128a$. It follows that the process stops after at most
      three steps (i.e., when $X$ has four elements), because after
      four steps we would have a bounding square of size at most 
      $2^5|S'|/128a \leq |S'|/4a$
      that is intersected by five disjoint squares of size at least 
      $|S'|/2^4 = |S'|/16$ (since 
      $T'$ is balanced), which is impossible (for $a$ large enough).
      Figure~\ref {fig:shift-rec} shows an example.

      Now we put two base squares around $\as_2(X)$: a square $\widetilde{R}$
      that is aligned with $T$, and a square $\widetilde{S}$ that is 
      aligned with $T'$. For $\widetilde{R}$, if $\as_2(X)$ contains 
      only one element, we just use the bounding square of $\as_2(X)$.
      If $|\as_2(X)| \geq 2$, then the elements
      of $\as_2(X)$ are separated by an edge or a corner between 
      leaf $T$-squares. Thus,
      we can pick a base square $\widetilde{R}$ for $\as_2(X)$ such that 
      (i) $|\widetilde{R}| \leq 2^6|S'|/128a = |S'|/2a$; (ii) $\widetilde{R}$ is 
      aligned with $T$; and (iii) the first split of $\widetilde{R}$ 
      separates the elements in $\as_2(X)$.
      For $\widetilde{S}$, if $|X| =1$, we just use the bounding square for
      $\as_2(X)$. If $|X| \geq 2$, the squares in $X$ must share a common
      edge or corner, and we can find a base square $\widetilde{S}$ such that
      (i) $\widetilde{S}$ contains $\as_2(X)$; (ii) the first split of 
      $\widetilde{S}$ produces squares that are aligned with this edge or 
      corner of $X$; and (iii) $|\widetilde{S}| \leq 2^6|S'|/128a = |S'|/2a$.
      Figure~\ref {fig:shift-rec2} shows an example.
      We now construct an $a$-compressed quadtree $\widetilde{T}$
      with base square $\widetilde{R}$ for the elements of $\as_2(X)$
      in the obvious way. (If $\as_2(X)$ contains any compressed children, 
      we reuse them as compressed children for $\widetilde{T}$. This may 
      lead to a violation of the condition for compressed nodes at the 
      first level of $\widetilde{T}$. However, our algorithm automatically 
      treats large compressed children as active squares, so there is no 
      problem.) This takes constant time.  We call the algorithm recursively 
      to shift $\widetilde{T}$ to the new base
      square $\widetilde{S}$. Note that this leads to a valid $a$-compressed
      quadtree since either $\widetilde{S}$ is wholly contained in $S'$;
      or the first split of $\widetilde{S}$ produces squares that are wholly
      contained in the $T'$-leaf squares and have size at most $|S'|/4a$,
      while each square that intersects $\widetilde{S}$ has size at least
      $|S'|/4$, as $T'$ is  balanced. We repeat the procedure 
      for every leaf $T'$-square whose secondary associated set we 
      have not processed yet.
      
      \paragraph{Analysis.}
      The resulting tree $T'$ is a balanced $a$-compressed quadtree
      for $P$. It remains to prove that the algorithm runs in linear time.
      The initialization stage needs $O(1)$ steps.
      Next, we consider the main body of the algorithm.
      Since each split takes constant time, the total running time 
      for the main body is proportional to the number of 
      splits. Recall that a $T'$-square $S'$ is called \emph{active} if 
      it is put into $L$, i.e., if $\as(S')$ contains a $T$-square of
      size in $[|S'|, 2|S'|)$ or a compressed child of size
      in $(|S'|/2^{2a}, |S'|]$. Since each $T$-square can cause only
      a constant number of $T'$-squares to be active, the total 
      number of active $T'$-squares is $O(m)$. Thus, we can use the following 
      lemma to conclude that the total number of
      splits in the main body of the algorithm is linear.
      
      \begin{lem}\label{lem:main-body}
         Every split in the main body of the algorithm can be charged
	 to an active $T'$-square such that each such square
	 is charged a constant number of times.
      \end{lem}

      \begin{proof}
      If we split an active square $S'$, we can trivially
      charge the split to $S'$.
      Hence, the critical splits
      are the ones during the balancing procedure. 
      By induction on the number of steps of the balancing
      procedure, we see that if a square $S'$ is split, there must be 
      a square $N'$ in the current partial tree $T'$ that is a direct 
      neighbor of $S'$ and that has an active descendant whose 
      removal from $L$  triggered the
      balancing procedure.\footnote{Recall that a direct neighbor
      of $S'$ is one of the eight squares of size $|S'|$ that surround
      $S'$.}

      If $N'$ has an active ancestor $\widetilde{N}$ that is at 
      most five levels above $N'$ in $T'$ (possibly $\widetilde{N} = N'$), 
      we charge the split of $S'$ to
      $\widetilde{N}$, and we are done. Otherwise, we know that
      $\as(N')$ contains at least one compressed child of
      size less than $|N'|/2^{2a}$ (otherwise, $N'$ would not
      have an active descendant or would itself be active) 
      and $T$-squares of size at least
      $64|N'|$ (otherwise, one of the five nodes above $N'$ in $T'$ 
      would have been active). Now, before $S'$ is split, there
      must have been a split on $N'$: otherwise the
      active descendant of $N'$ that triggers
      the split on $S'$ would not exist. Thus, we repeat the argument to show
      that $N'$ has a direct neighbor $N''$ with an
      active descendant that triggers the split of $S'$. 
      Note that $N'' \neq S'$, because the split on $N'$ happens before
      the split on $S'$.
      If $N''$ has an active ancestor that is at most five levels higher up in
      $T'$ (possibly $N''$ itself), we are done again. 
      Otherwise, we repeat the argument again. 
      
      We claim that this process finishes after at most $16$ steps. 
      Indeed, suppose we find $17$ squares 
      $S' = N^{(0)}, N^{(1)}, N^{(2)}, \ldots, N^{(17)}$ without stopping. 
      We know that each $N^{(j)}$ is a direct neighbor of $N^{(j-1)}$ and that
      each $N^{(j)}$ is associated with a compressed child of size
      at most $|S'|/2^{2a}$ and with $T$-squares of size at least
      $64|S'|$. Since the set $\bigcup_{j=0}^{16} N^{(j)}$ has
      diameter at most $17|S'|$, the set $\bigcup_{j=0}^{16} \as(N^{(j)})$
      contains at most four $T$-squares of size at least $64|S'|$.
      Now each compressed child in an associated set
      $\as(N^{(j)})$ is the only child of one of these four large $T$-squares,
      so there are at most four of them.
      Furthermore, each such compressed child is intersected by at
      most four disjoint $T$-squares of size $|S'|$, so 
      there can be at most $16$ squares $N^{(j)}$, a contradiction.
      Hence, we can charge each split to an active square in the
      desired fashion, and the lemma follows.
      \end{proof}
       
      Next, we analyze the running time of the secondary stage. Again,
      the running time is proportional to the number of splits, which
      is bounded by the following lemma.

      \begin{lem}\label{lem:secondary-bound}
         Let $S'$ be a frontier $T'$-square at the beginning of the
	 secondary stage. Then after the secondary stage, the subtree
	 rooted at $S'$ has height at most $O(\log a)$.
      \end{lem}
      \begin{proof}
      Below, we will argue that for every descendant $S''$ of $S'$,
      if $\as_2(S'')$ contains a compressed child $C$, then
      $|C| \leq |S''|/2^a$. For now, suppose that this holds.

      First, we claim that there are $O(\log a)$ splits
      to post-active descendants of $S'$. 
      The secondary associated set $\as_2(S')$ contains at most four 
      elements, so $\as_2(S')$ has at most $11$ subsets with
      two or more elements. Fix such a subset $\mathcal{A}$.
      Then $S'$ has at most $O(\log a)$ post-active
      descendants with secondary associated set $\mathcal{A}$.
      This is because each level of $T'$ has at most two squares
      with secondary associated set $\mathcal{A}$, and 
      the post-active squares with secondary associated set $\mathcal{A}$
      must have size between $|B(\mathcal{A})|/2$ and 
      $128a|B(\mathcal{A})|$, where $B(\mathcal{A})$ denotes the
      smallest bounding square for the elements in $\mathcal{A}$. 
      (Here we use our claim that the compressed children in
      the secondary associated set of each frontier $T'$-square 
      $S''$ are much smaller
      than $S''$.) There are only $O(\log a)$ such levels, so
      adding over all $\mathcal{A}$, we see that $S'$ has 
      at most $O(\log a)$ post-active descendants, implying the
      claim.

      Each split creates at most one new level below $S'$,
      so there are only $O(\log a)$ new levels 
      due to splits to post-active descendants of $S'$.
      Next, we bound the number of new levels that are created by
      splits during the balancing phases. Each balancing phase
      creates at most one new level below $S'$.
      Furthermore, by induction on the number of steps in the 
      balancing phase, we see that the balancing phase was
      triggered by the split of a post-active square
      that is a descendant either of $S'$ or of a direct neighbor
      of $S'$.
      At the beginning of the secondary stage, there
      are $O(1)$  $T'$-squares that are descendants of direct 
      neighbors of $S'$ (as $T'$ is balanced). As we argued above, each 
      of them has at most $O(\log a)$ post-active descendants.
      Thus, the balancing phases add at most $O(\log a)$ new levels below
      $S'$.
      
      Finally, we need to justify the assumption that for any descendant $S''$
      with a compressed child $C \in \as_2(S'')$, we have $|C| \leq |S''|/2^a$.
      By construction, we have $|C| \leq |S'|/2^{2a}$. Suppose that $S'$
      has a descendant $S''$ that violates this assumption. The square $S''$ was
      created through a split in the secondary stage, and suppose that $S''$
      is the first such square during the whole secondary stage. This means
      that during all previous splits, the assumption holds, so by the
      argument above, there are at most $O(\log a)$ levels below $S'$. 
      This means that $|S''| \geq |S'|/a^{O(1)}$, so we would get
      $|S'|/2^{2a} \geq |C| > |S''|/2^a > |S'|/2^{2a}$,
      a contradiction (for $a$ large enough). Thus, no $S''$ can violate the
      assumption, as desired.
      \end{proof}
       
      The time to set up the recursion is constant for each square of the
      current frontier. From Lemmas~\ref{lem:main-body} and 
      \ref{lem:secondary-bound}, we can conclude that the total time of
      the algorithm is $O(m)$, which also implies that $T'$ has $O(m)$
      squares. This concludes the proof of Theorem~\ref{thm:shiftbalance}.

    \paragraph{Special Cases.}
    We note two useful special cases of Theorem~\ref{thm:shiftbalance}.
    The first one gives an analog of Theorem~\ref{thm:balance} for compressed
    quadtrees.
    \begin {cor} \label {cor:balance-'n-thread}
      Let $T$ be a $a$-compressed quadtree with $m$ nodes. 
      There exists a balanced $a$-compressed quadtree that contains $T$,
      has $O (m)$ nodes and can be constructed in $O (m)$ time.
    \end {cor}

   \begin{proof}
      Let $R$ be the base square of $T$. We apply 
      Theorem~\ref{thm:shiftbalance} with $S = R$. 
   \end{proof}

    The second special case says that we can realign an uncompressed 
    quadtree locally in any way we want, as long as we are willing
    to relax the definition of quadtree slightly.\footnote
    {We cannot get a non-relaxed ($1$-relaxed) uncompressed quadtree, since two points could be arbitrarily close to each other if they were separated by a boundary. However, we can always turn a $\lambda$-relaxed quadtree into a non-relaxed compressed quadtree in linear time again.}
    Let $P$ be a planar point set.
    We call a quadtree for $P$ \emph {$\lambda$-relaxed} if it has at most
    $\lambda$ points of $P$ in each leaf, and is otherwise a regular quadtree.
    \begin {cor} \label {cor:qt-shift}
      Let $P$ be a planar point set and $T$ a regular quadtree for $P$, 
      with base square $R$.
      Let $S$ be another square with $S \supseteq P$ and 
      $|S| = \Theta(|R|)$.
      Then we can build a $4$-relaxed quadtree $T'$ for $P$ 
      with base square $S$ in $O(|T|)$ time such that
      $T'$ has $O(|T|)$ nodes.      
    \end {cor}

   \begin{proof}
      We apply Theorem~\ref{thm:shiftbalance} to $T$, but we stop the
      algorithm before the beginning of the secondary stage.
      Since each secondary associated set for a leaf square has
      at most four elements, and since $T$ contains no compressed
      nodes, the resulting tree $T'$ has the desired properties.
   \end{proof}

\subsection{Equivalence of Compressed and $c$-Cluster Quadtrees}
\label{sec:compressed-c-cluster}

    The goal of this section is to prove the following theorem.
    \begin{theorem}\label{thm:cluster-compressed-equiv}
      Let $P$ be a planar $n$-point set.  Given a
      $(c_1, c_2)$-cluster quadtree on $P$, we can compute in $O(n)$ time an
      $O(c_1)$-compressed quadtree on $P$; and given an $a$-compressed
      quadtree on $P$, we can compute in $O(n)$ time an
      $(a^{1/5}, 2a^{1/5})$-cluster quadtree on $P$.
    \end{theorem}
   
    We present the proof of Theorem~\ref{thm:cluster-compressed-equiv}
    in two lemmas.

    \begin{lem} \label{lem:c-cluster->compressed}
    Let $P$ be a planar $n$-point set. Given a
    $(c_1, c_2)$-cluster quadtree $T$ for $P$, we can compute in linear time an
    $O(c_1)$-compressed quadtree $T'$ on $P$.
    \end{lem}

    \begin{proof} 
      We construct the compressed quadtree in a top-down fashion, 
      beginning from the root. Suppose that we have constructed a
      partial compressed quadtree $T'$, and let $q$ be the representative
      point for a node $u$ in the $(c_1,c_2)$-cluster tree 
      $T_{(c_1,c_2)}$ that corresponds
      to $T$. We show how to expand $q$ in $T'$ to the corresponding
      quadtree $T_u^Q$. 

      First, we add to $T_u^Q$ a new root that is aligned with the
      old base square and larger by a constant factor, such that
      the old base square does not touch any boundary of the
      new one. Next, we determine by a search from $q$
      which leaf squares of $T'$ intersect
      $T_u^Q$. By Observation~\ref{obs:close}, there are at most
      four such leaves, so this step takes constant time.
      (Note that since we grow the base square of each quadtree that we
      expand, it cannot happen that $T_u^Q$ intersects the boundary of
      its parent quadtree.) Next, we repeatedly split
      each leaf that intersects $T_u^Q$ and that contains some other point
      or compressed child until there are no more such leaves.

      The proof of Observation~\ref{obs:close} shows that
      every leaf square of $T'$ that intersects
      $T_u^Q$ has size at least $c_1d/4$, where $d$
      is the size of $T_u^Q$'s base square.
      If $T_u^Q$ lies completely inside a
      leaf of $T'$, we add $T_u^Q$ as a compressed child to $T'$. 
      If $T_u^Q$ intersects more than one leaf square,  we
      identify a square at most twice the size of $T_u^Q$'s base square that is
      aligned appropriately with the relevant edges of $T$, 
       and apply Corollary~\ref{cor:qt-shift} to shift $T_u^{Q}$ 
       to this new base square.  
      This results in a valid $O(c_1)$ compressed quadtree in which
      $q$ has been expanded. We repeat this process until all the quadtree
      pieces of $T$ have been integrated into a large compressed quadtree.

      The total time for the top-down traversal and for the realignment
      procedures is linear. Furthermore, 
      Corollary~\ref{cor:c-cluster-QT-extended} shows that the 
      total work for splitting
      the leaves of $T'$ is also linear, since the points in the different
      clusters are $(1/c_1)$-semi-separated. Hence, the total running time
      is linear.
    \end{proof}

    \begin{lem}\label{lem:compressed->c-cluster}
      Let $P$ be a planar $n$-point set, and $T$ be an $a$-compressed
      quadtree for $P$. Then we can compute in linear time a
      $(a^{1/5}, 2a^{1/5})$-cluster quadtree for $P$.
    \end{lem}

    \begin{proof} 
      We use 
      Corollary~\ref{cor:balance-'n-thread} to balance $T$, 
      but without the recursive calls for the remaining cluster
      nodes. This gives a balanced top-level quadtree $T_\text{top}$ 
      (possibly with some compressed children of $T$ now 
      integrated in the tree), in which each leaf square is  
      associated with at most four points from $P$ or compressed
      children of $T$. Furthermore, for each leaf square $S$ of $T_\text{top}$,
      we have a bounding square for the associated elements
      that is aligned with $T$ and has size at most 
      $|S|/a$. 

      We use $T_{\text{top}}$ to identify a partial 
      cluster quadtree, and we then recurse on the 
      compressed children.
      We say a square $S \in T_\text{top}$ is \emph{full} if 
      there is a leaf below $S$ with a non-empty associated set.
      Otherwise, $S$ is \emph{empty}.
      First, we consider the squares of $T_\text{top}$ in top-down
      fashion and check for each
      full square $S$ which direct neighbors of $S$ 
      are empty  (this can
      be done in constant time since $T$ is balanced). If $S$ has
      at most three full direct neighbors, and if all these
      full squares share a common corner,
      we let $U$ be a square that
      is aligned with $S$ and contains the full squares (i.e., either
      $U = S$ or $U$ is a square of size $2|S|$ that contains $S$ and
      its full neighbors). Next, we 
      consider the squares of size $|U|$ in
      the $(4a^{1/5}+1) \times (4a^{1/5}+1)$ grid centered at $U$
      and check whether they are all empty (again, since $T$ is balanced,
      this takes constant time).
      If so, the points associated with $U$
      define a $a^{1/5}$-cluster. We put a representative
      point for the cluster into $U$, make a new quadtree with root $U$,
      and remove $U$'s children from $T_\text{top}$.
      We continue until all the squares of $T_\text{top}$ have been traversed,
      and then we process all the new trees in a similar way, iterating
      if necessary. After we are done, a part of the cluster quadtree has
      been created, and we need to consider the compressed children
      to set up a recursion.

      For this, we consider each non-empty leaf square $S$ of the partial tree.
      Let $B$ be the bounding square of the associated elements of $S$.
      We know that $|B| \leq |S|/a$, so the disc $D$ of radius  $2|B|a^{1/5}$
      centered at $B$ intersects at most three other leaf squares. 
      We check for each of these leaf squares whether $D$ intersects
      the bounding square of its associated elements.
      If so, we make a new bounding square for the union of these elements and
      repeat. This can happen at most twice more, because in each step
      the size of the bounding square increases by a factor of at most $a^{1/5}$.
      Hence, after three steps we have a disk $D$ of 
      radius $O(|B|a^{4/5})$ that 
      intersects four disjoint squares of size $\Omega(|B|a)$ that share a
      corner. Thus, $D$ must be completely contained in those squares. 
      This also implies that this procedure yields a
      $a^{1/5}$-cluster.  For each such cluster, we 
      create a representative point and
      an appropriate base square for the child quadtree. 
      Then, we process the cluster
      recursively. In the end, we can prune the resulting compressed
      trees to remove unnecessary nodes.

      By the proof of Corollary~\ref{cor:balance-'n-thread},
      and since be spend only constant additional time for each square,
      this procedure takes linear time. Furthermore, as we argued
      above, we create only $a^{1/5}$-clusters. If $Q \subset P$
      is a $2a^{1/5}$-cluster, then $Q$ is either contained in
      at most four leaf squares of $T_\text{top}$ that share a corner or
      the bounding square $B_Q$ intersects at most four squares of 
      $T_\text{top}$ of size $\Theta(|B_Q|)$ such that the
      surrounding $(4a^{1/5}+1) \times (4a^{1/5}+1)$ grid contains only  empty
      squares. In either case, $Q$ (or a superset) is discovered.
      It follows that the result is a valid $(a^{1/5},2a^{1/5})$-cluster 
      quadtree.
    \end{proof}

\section {From a $c$-Cluster Quadtree to the Delaunay Triangulation}
\label{sec:qt->dt}

We now come to the heart of the matter and show how to construct a DT from
a WSPD. Let $P$ be a set of points, and $T$ a
compressed quadtree for $P$. Throughout this section, $\eps$ is a
small enough constant (say, $\eps = \pi/400$), and $k$ is
a large enough constant (e.g., $k = 100$).
Let $u$ and $v$ be two \emph{unrelated} nodes of $T$, i.e.,
neither node is an ancestor of the other. Let $L_{uv}$ be the 
set of directed lines that stab $S_u$ before $S_v$. 
The set $\Phi_{uv} \subseteq [0, 2 \pi)$ of directions for $L_{uv}$
is an interval modulo ${2\pi}$
whose extreme points correspond to
the two diagonal bitangents of $S_u$ and $S_v$, i.e., the two lines that
meet $S_u$ and $S_v$ in exactly one point each and have $S_u$ and
$S_v$ to different sides. 
Figure~\ref {fig:angle-int} illustrates this.

\tweeplaatjes {angle-int} {angle-sep} {(a) The set of possible directions 
between two unrelated nodes $u$ and $v$. (b) The set of possible directions 
between well-separated pairs is small.}

\begin {observation} \label {obs:directions-contained}
Let $u$ and $v$ be two unrelated nodes of $T$, and let $\child u$ be a
descendant of $u$ and $\child v$ be a descendant of $v$.
Then $\Phi_{\child u \child v} \subseteq \Phi_{uv}$.
\end {observation}

\begin{proof}
This is immediate, because $S_{\child u} \subseteq S_u$ and
$S_{\child v} \subseteq S_v$.
\end{proof}

\begin{observation}\label{obs:small-Phi}
If $u$ and $v$ are two  nodes of $T$ such that $\{u,v\}$ is
$\eps$-well-separated, then
$|\Phi_{uv}| \leq 8\eps$.
\end{observation}

\begin{proof}
  Let $d \eqdef |c_uc_v|$, $D_u$ be the disk around $c_u$ with radius
  $\eps d$, and $D_v$ the disk around $c_v$ with the same radius.\footnote
  {Recall, $c_u$ is the center point of $B_u$.}
  By well-separation, $S_u \subseteq D_u$ and $S_v \subseteq D_v$. Let $\beta$
  be the angle between the 
  diagonal bitangents of $D_u$ and $D_v$. Then
  $|\Phi_{uv}| \leq \beta$, and 
  $ 
  \beta = 2 \arcsin (\eps d / \frac12 d)
  = 2 \arcsin (2 \eps) \leq 8 \eps,
  $ 
  as claimed. 
  Figure~\ref {fig:angle-sep} illustrates this.
\end{proof}

For a number  $\phi \in [0, 2 \pi[$ we define
$\Phi_\phi \eqdef \{\psi \bmod 2 \pi \mid
\psi \in [\phi - \eps/2,
\phi + \eps/2] \}$, i.e., the set
of all directions that differ from $\phi$ by at most
$\eps/2$. We say
that an ordered pair  $(u, v)$ of nodes has direction $\phi$ if
$\Phi_{uv} \cap \Phi_\phi \neq \emptyset$.
We also say that a pair of points $(p,q)$
has direction
$\phi$ if the corresponding pair in the WSPD has
direction $\phi$. The same definition also applies to an edge.
For a given point $p$ in the plane, we define the
$\eps$-cone $\mathcal{C}_\phi (p)$ as the cone with apex $p$ and
opening angle $\eps$ centered around the direction $\phi$.

\subsection {Constructing a Supergraph of the EMST}\label{sec:emstsup}

In the following, we abbreviate $\mathcal{P} \eqdef \wspd(T)$.
The goal of this section is to construct a graph $H$
with vertex set $P$ and $O(n)$ edges, such that
$\emst(P) \subseteq H$. 
It is well known that if we take the graph $H'$ on $P$ with edge set
$E \eqdef \{e_{uv} \mid \{u,v\} \in \mathcal{P}\}$, where
each $e_{uv}$ connects the bichromatic closest pair for $P_u$ and $P_v$, 
then $H'$ contains $\emst(P)$ and has $O(n)$ edges~\cite{Eppstein00}.
However, as defined, it is not clear how to find $H'$ in linear time.
There are several major obstacles. Firstly, even though the tree $T$
has $O(n)$ nodes, it could be that $\sum_{u \in T} |P_u| = \Omega(n^2)$.
Secondly, even if the total size of all $P_u$'s was $O(n)$, we still need
to find bichromatic closest pairs for all \emph{pairs} in $\mathcal{P}$.
Thus, a large set $P_u$ might appear in many pairs of
$\mathcal{P}$, making the total problem size superlinear. Thirdly,
we need to actually solve the bichromatic closest pair problems. A 
straightforward solution to find the bichromatic closest pair for
sets $R$ and $B$ with sizes $r$ and $b$ would take time 
$O((r + b) \log(\min(r,b))$, by computing the Voronoi diagram
for the smaller set and locating all points from the other set in it.
We need to find a way to do it in linear time.

To address these problems, we actually construct a slightly larger graph
$H$, by partitioning the pairs in $\mathcal{P}$ according to their
direction. More precisely,
let $Y = \{0, \eps, 2\eps, \ldots, (l-1)\eps\}$ be a set of
$l$ numbers, where we assume 
that $l = 2\pi/\eps$ is an integer.
For every  $\phi \in Y$, we  construct a graph $H_\phi$
with $O(n)$ edges and then let $H = \bigcup_{\phi \in Y} H_\phi$.
Given $\phi \in Y$, the graph $H_\phi$ is constructed in three steps:
\begin{enumerate}
\item 
  For every node
  $u \in T$, select a subset $Z_u \subseteq P_u$, such that
  $\sum_{u \in T} |Z_u| = O(n)$,
  and such that 
  $\{\{p, q\} \mid p \in Z_u, q \in Z_v, \{u, v\} \in \mathcal P\}$ still
  contains all edges of $\emst(P)$ with orientation $\phi$. This addresses
  the first problem by making the total set size linear.
\item 
  Find a subset $\mathcal{P}' \subseteq \mathcal{P}$,
  such that
  each $u \in T$ appears in $O(1)$ pairs of
  $\mathcal{P}'$, and the set 
  $\{\{p, q\} \mid p \in Z_u, q \in Z_v, \{u, v\} \in \mathcal {P}'\}$
  contains all edges of $\emst(P)$ with orientation $\phi$.
  In particular, we choose
  for every node $u \in T$ a subset
  $\mathcal{P}_u \subseteq \mathcal {P}$ such that
  $\mathcal{P}' = \bigcup_{u  \in T} \mathcal{P}_u$,
  each  pair in $\mathcal{P}_u$ contains $u$, and
  $|\mathcal{P}_u| = O(1)$. This addresses the second problem
  by ensuring that every set appears in $O(1)$ pairs.
\item 
  For every pair $\{u, v\} \in \mathcal{P}'$, we
  include in $H_\phi$ the edge $pq$ such that
  $\{p,q\}$ is the closest pair in $Z_u \otimes Z_v$ (i.e.,
  $\{p,q\} = \argmin_{\{p',q'\} \in Z_u \otimes Z_v} |p'q'|$).
  Here we actually solve all the bichromatic closest pair problems.
\end{enumerate}
Clearly, $H_\phi$ has $O(n)$ edges, and we will show that $H$ is indeed a
supergraph of $\emst(P)$. Our strategy of subdividing the edges
according to their orientation goes back to Yao, who
used a similar scheme to find EMSTs in higher dimensions~\cite{Yao82}. 

\paragraph{Step 1: Finding the $Z_u$'s.}
Recall that we fixed a direction $\phi \in Y$.  Take the set
$\mathcal{P}_\phi \subseteq \wspd(T)$ of pairs with direction
$\phi$. For a pair $\pi \in \mathcal{P}_{\phi}$, we write 
$(u,v)$ for the tuple such that $\pi = \{u,v\}$ and $c_u$ comes before
$c_v$ in direction $\phi$, it is a \emph{directed} pair in $\mathcal{P}_\phi$. 
Call a node $u$ of $T$ \emph{full} if either 
(i) $u$ is the
root; (ii) $u$ is a non-empty leaf; or (iii) $\mathcal{P_\phi}$ has a
directed pair $(u, v)$. 
Let $T'$ be the tree 
obtained from $T$ by
connecting every full node to its closest full ancestor, and by
removing the other nodes. We can compute $T'$ in linear time
through a post-order traversal.  Now, for every leaf $v$ of $T'$, put
the point $p \in P_v$ into the sets $Z_u$, where $u$ is one the $k$\footnote
{Recall, $k$ is a sufficiently large constant.}
closest ancestors of $v$ in $T'$. Repeat this procedure, while changing
property (iii) above so that $\mathcal{P}_\phi$ has a directed pair
$(v,u)$. 
This takes
linear time, and $\sum_{u \in T} |Z_u| = O(n)$. Intuitively, $Z_u$
contains those points of $P_u$ that are sufficiently on the outside
of the point set in direction $\phi$.
Figure~\ref {fig:qt-example+qt-wspd+qt-zu} shows an example.
\drieplaatjesbreed {qt-example} {qt-wspd} {qt-zu}
{(a) A node $u$ in the quadtree, with $|P_u| = 8$.
 (b) The relevant wspd-pairs (in green) for the points in $P_u$
     with direction $\phi$ (up). There are also wspd-pairs between $u$ and
     other nodes above and below it.
 (c) For $k=1$, $Z_u$ contains those $p \in P_u$ for which the 
     lowest wspd-pair in the tree $T'$ that involves $p$ 
     contains $u$.
     In other words, $Z_u$  has the points that do not have a green edge in
     both directions in (b).
}
Variants of the following claim
have appeared several times before~\cite{AgarwalEdScWe91,Yao82}.

\begin{claim}\label{clm:lune_nn}
Let $p \in P$, and let $\mathcal{C}^+_{\phi}(p)$ denote the cone with
apex $p$ and opening angle $17\eps$ centered around $\phi$.
Suppose that $pq$ is an edge of $\emst(P)$ and
$q \in \mathcal{C}^+_{\phi}(p)$.  Then $q$ is the nearest neighbor of $p$
in $\mathcal{C}^+_{\phi}(p) \cap P$.
\end{claim}

\begin{proof}
If $pq$ is an edge of $\emst(P)$, then the \emph{lune} $L$ defined by
$p$ and $q$ contains no point of $P$~\cite{deBergChvKrOv08}.\footnote{$L$
is the intersection
of two disks with radius $|pq|$, one centered at $p$, the other centered
at $q$.}
Since the opening angle of $\mathcal{C}^+_{\phi}(p)$ is at most  $\pi/3$,
for $\eps$ small enough,
the intersection of $\mathcal{C}^+_{\phi}(p)$ with $L$ equals
the intersection of $\mathcal{C}^+_{\phi}(p)$ with the disk around 
$p$ of radius $|pq|$. Hence, $q$ must be the nearest neighbor of
$p$ in $\mathcal{C}^+_{\phi}(p) \cap P$.
\end{proof}

\begin{lem}\label{lem:constant-levels}
Let $pq$ be an edge of $\emst(P)$ with direction $\phi$, and let
$\{u,v\}$ be the corresponding wspd-pair.
Then $\{p,q\} \in Z_u \otimes Z_v$.
\end{lem}

\begin{proof}
Let $w$ be the leaf for $p$, and suppose for contradiction
that $p \notin Z_u$, i.e., $u$ is not among the  $k$ closest
ancestors of $w$ in $T'$.
This means there exists a
sequence
$u_1, u_2, \ldots, u_k, u$ of $k+1$ distinct ancestors of $w$,
such that each node is an ancestor of all previous nodes and such
there are well-separated
pairs $\{u_1, v_1\}, \{u_2, v_2\}, \ldots, \{u_k, v_k\} \in \mathcal{P}_\phi$.

Let $\mathcal{C}^+_{\phi}(p)$ be the cone with apex $p$ and opening angle
$17\eps$ centered around $\phi$.  By Observation~\ref{obs:small-Phi},
we have $S_v, S_{v_1}, \ldots, S_{v_k} \subseteq \mathcal{C}^+_{\phi}(p)$.
Furthermore, since $\{u, v\}$ is well-separated, 
$d(u, v) \geq |S_u|/\eps$.
Now Claim~\ref{clm:quad-squares} implies that there are
squares $R_{u_1}$, $R_{v_1}$ such that
(i) $S_{u_1} \subseteq R_{u_1} \subseteq S_{u_2}$
and  $S_{v_1} \subseteq R_{v_1}$;
(ii) $|R_{u_1}| = |R_{v_1}|$; and
(iii) $d(R_{u_1}, R_{v_1}) \leq 2|R_{u_1}|/\eps$.
This means that 
\[
d(p, P_{v_1}) \leq 2(1+1/\eps)|R_{u_1}| \leq
2(1+1/\eps)|S_{u_2}| \leq 2(1+1/\eps)|S_u|/2^{k-1},
\]
where in the first inequality we bounded the distance between
any point in $R_{u_1}$ and any point in $R_{v_1}$ by the distance between the
squares plus their diameter (since we do not know where
the points lie inside the squares). The second inequality
comes from $R_{u_1} \subseteq S_{u_2}$ and 
the third inequality is due to the fact that $S_{u_2}$ lies
at least $k-1$ levels below $S_u$ in $T'$.

Since $2(1+1/\eps)/2^{k-1} < 1/\eps$ for $k \geq 3$ and  since
$d(u,v) \geq |S_u|/\eps$,
this contradicts the fact that $q$ is the nearest neighbor
of $p$ inside $\mathcal{C}^+_{\phi}(p)$ (Claim~\ref{clm:lune_nn}). 
Thus, $p$ must lie in $Z_u$.
A symmetric argument shows $q \in Z_v$.
\end{proof}

\paragraph{Step 2: Finding the $\mathcal{P}_u$'s.}
For every node $u \in T$, we include in $\mathcal{P}_u$ the $k$ shortest
pairs in direction $\phi$, i.e., the
pairs $\{u, v\} \in \wspd(T)$ such that
(i) $c_v$ is contained
in the $\eps$-cone $\mathcal{C}_\phi(c_u)$ with apex $c_u$
centered around direction $\phi$;
and (ii) there are less than $k$ pairs $\{u, v'\} \in \wspd(T)$
that fulfill (i) and have
$|c_uc_{v'}| < |c_uc_v|$.
Since $k$ is constant,
the $\mathcal{P}_u$'s can be constructed in total linear
time.
Even though each $\mathcal{P}_u$ contains a constant number of
elements, a node might still appear in many such sets, so
we further prune the pairs:
by examining the $\mathcal{P}_u$'s, determine for each $v \in T$ the
set $\mathcal{Q}_v = \{u \in T \mid v \in \mathcal{P}_u\}$.
For each $\mathcal{Q}_v$, find the $k$ closest neighbors
(measured by the distance between their center points) of $v$ in
$\mathcal{Q}_v$, and for all other $\mathcal{P}_u$'s remove the
corresponding pairs $\{u,v\}$.
Now each node appears in only a constant number of pairs of
$\mathcal{P}' = \bigcup_{u  \in T} \mathcal{P}_u$.

\begin {lem}\label{lem:constant-neighbors}
Let $pq$ be an edge of $\emst(P)$ with orientation $\phi$, and
let $\{u,v\}$ be the corresponding wspd-pair.
Then $\{u,v\} \in \mathcal{P}_u$.
\end {lem}

\begin {proof}
We show that $v$ is among the $k$ closest neighbors
of $u$ in direction $\phi$, a symmetric argument shows that $u$
is among the $k$ closest neighbors of $v$ in direction $-\phi$.
We may assume that  $|c_uc_v| =1$.
Suppose that $\{u,v\}$ is not among the $k$ shortest pairs
in direction $\phi$.
Then there is a set $W$ of $k$ nodes
of $T$ such that for all $w \in W$ we have
(i) $c_w \in \mathcal{C}_\phi(c_u)$;
(ii) $|c_uc_w| < 1$; and
(iii) $\{u,w\} \in \wspd(T)$.
By Claim~\ref{clm:quad-squares}, there exists for every
$w \in W$ a pair of squares $R_u(w), R_w$ such that
$S_u \subseteq R_u(w)$, $S_w \subseteq R_w$ and
$|R_u(w)| = |R_w| \leq 2\eps d(R_u(w), R_w) \leq 2\eps$.

\eenplaatje[scale=0.9]{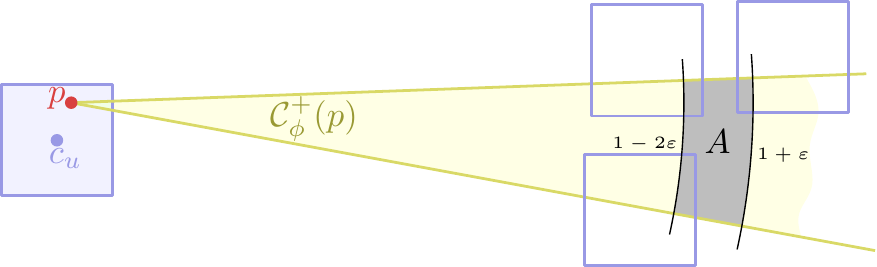} {All squares $R_w$ intersect the region $A$.}

Let $\mathcal{C}^+_{\phi}(p)$ be the cone with apex $p$ and opening angle
$17\eps$ centered around $\phi$. By Observation~\ref{obs:small-Phi},
$S_w \subseteq \mathcal{C}^+_{\phi}(p)$ for all $w \in W$.
Furthermore, every $S_w$ contains a point at distance at most $1+\eps$ from
$p$, because $|c_wp| \leq |c_wc_u| + |c_up| \leq 1+ \eps$.
Also, by Claim~\ref{clm:lune_nn}, every $S_w$ contains a point
at distance at least $|pq| \geq |c_u c_v| - |c_u p| - |q c_v| \geq 1-2\eps$
from $p$. Thus, since $d(R_u(w), R_w) \leq 2|R_w|/\eps$ 
by Claim~\ref {clm:quad-squares} and
$d(R_u(w), R_w) \geq 1 - 2\eps -2|R_w|$, we get $|R_w| \geq \eps/8$,
for $\eps$ small enough. However, this implies that
$W$ has only a constant number of squares:
all $S_w$ (and hence all $R_w$)
intersect the annular segment $A$  inside
$\mathcal{C}^+_{\phi}(p)$ with inner radius $1-2\eps$ and outer radius
$1+\eps$ (see Figure~\ref{fig:lune}).
All $w \in W$ are unrelated, since they are paired with
$u$ in $\wspd(T)$. Furthermore, the set $A$ has diameter $O(\eps)$.
If $w \in W$ is a compressed child, then $R_w$ is contained in the
parent of $w$ and intersects no
other $S_{w'}$, for $w' \in W$. Otherwise, 
$|S_w| \geq |R_w|/2$. Thus, if we assign to each compressed child
$w \in W$ the square $R_w$ and to each other node $w \in W$ the square
$S_w$, we get a collection of $k$ disjoint squares that
meet $A$ and each have diameter $\Omega(\eps)$. 
Since $A$ has diameter $O(\eps)$, 
there can be only a constant number of such
squares, so
choosing $k$ large enough leads to a contradiction.
\end{proof}

\paragraph{Step 3: Finding the Nearest Neighbors.}
Unlike in the previous steps, the algorithm for Step 3 is a bit
involved, so we switch the order and begin by
showing correctness.

\begin{lem}\label{lem:emst_nn}
Let $pq$ be an edge of $\emst(P)$ with direction
$\phi$ and let $\{u,v\}$ be the corresponding wspd-pair.
Then $\{p,q\}$ is the closest pair in $Z_u \otimes Z_v$.
\end{lem}

\begin{proof}
By Lemma~\ref{lem:constant-levels}, we have $\{p,q\} \in Z_u \otimes Z_v$.
Furthermore, the cut property of minimum spanning trees implies that
$pq \in \emst(Z_u \cup Z_v)$.
Since $\{u, v\}$ is well-separated, we have
\begin{equation}\label{equ:order}
\max_{\{p',q'\} \in Z_u \otimes Z_u \cup Z_v \otimes Z_v} |p'q'| <
\min_{\{p',q'\} \in Z_u \otimes Z_v} |p'q'|.
\end{equation}
Now consider an execution of Kruskal's MST algorithm on
$Z_u \cup Z_v$~\cite[Chapter~23.2]{CormenLeRiSt09}.
Let $\{p', q'\}$ be the closest pair
in $Z_u \otimes Z_v$. By $(\ref{equ:order})$, the algorithm
considers $p'q'$ only after processing all edges in
$Z_u \otimes Z_u \cup Z_v \otimes Z_v$. Hence, at that point
the sets $Z_u$ and $Z_v$ are each contained in a connected component
of the partial spanning tree, and $\emst(Z_u \cup Z_v)$ can have at most one
edge from $Z_u \otimes Z_v$. Hence, it follows that
$\{p,q\} = \{p',q'\}$, as claimed.
\end{proof}

We now describe the algorithm.
For ease of exposition, we take $\phi = \pi/2$ (i.e.,
we assume that $P$ is rotated so that $\phi$ points in the positive
$y$-direction).
Note that now the squares are not generally axis-aligned anymore, but this
will be no problem.
Given a point $p \in \R^2$, we define the four
\emph{directional cones}
$\mathcal{C}_\leftarrow(p)$,$\mathcal{C}_\uparrow(p)$,
$\mathcal{C}_\rightarrow(p)$, and $\mathcal{C}_\downarrow(p)$
as the leftward, upward, rightward and downward cones with apex $p$ and
opening angle $\pi/2$. The directional cones subdivide the plane into
four disjoint sectors. We will also need the \emph{extended} rightward
cone $\mathcal{C}_\rightarrow^+(p)$ with apex $p$ and opening
angle $\pi/2+16\eps$.

\begin{claim}\label{clm:emptycone}
Let $(u,v)$ be a directed pair in  $\mathcal{P}_\phi$, 
and suppose that $\{p,q\}$ with $p \in P_u$ and $q \in P_v$
is the closest pair for $(u,v)$. 
Then $\mathcal{C}_\uparrow(p) \cap P_u = \emptyset$ and
$\mathcal{C}_\downarrow(q) \cap P_v = \emptyset$.\footnote{Recall 
that we set $\phi=\pi/2$,
so $\uparrow$ and $\downarrow$ mean ``in direction $\phi$''
and ``in direction $-\phi$''.}
\end{claim}

\eenplaatje {closest-cone} {The intersection points of $D$ and the boundary of
$\mathcal{C}_\downarrow(q)$ lie outside $S_v$, so
$S_v \cap \mathcal{C}_\downarrow(q) \subseteq D$.}

\begin{proof}
We prove the claim for $\mathcal{C}_\downarrow(q)$, the argument for
$\mathcal{C}_\uparrow(p)$ is symmetric. We may assume that $|pq| = 1$.
By assumption, the unit disk $D$ centered at $p$ contains no points of $P_v$,
so it suffices to show that
$\mathcal{C}_\downarrow(q) \cap S_v \subseteq D$.
Since $\{u,v\} \in \mathcal{P}_\phi$
and by Observation~\ref{obs:small-Phi}, the direction of the
line $\overline{pq}$ differs from $\phi$ by at most $17\eps$.
Therefore, the intersections of the boundaries of
$\mathcal{C}_\downarrow(q)$ and $D$
have distance at least $\sqrt{2} - O(\eps)$ from $q$. However, the
pair $\{u,v\}$ is well-separated, so all points in $P_v$ have distance at
most $\eps$ from $q$, which implies the claim;
see Figure~\ref{fig:closest-cone}.
\end{proof}

Given a  set $Z_u$ for a node $u$ of $T$, we define the \emph{upper chain}
of $Z_u$, $\UC(Z_u)$ as follows:
remove from $Z_u$ all points $p$ such that
$\mathcal{C}_\uparrow(p)$ contains a point from $Z_u$
in its interior. Then sort $Z_u$ by $x$-coordinate and connect consecutive
points by 
line segments.
All segments of $\UC(Z_u)$ have slopes in $[-1,1]$.
Similarly, we define the \emph{lower chain} of $Z_u$, $\LC(Z_u)$, by requiring
the cones $\mathcal{C}_\downarrow(p)$  for the points in $\LC(Z_u)$ to be
empty. The goal now is to compute $\UC(Z_u)$ and $\LC(Z_u)$ for all
nodes $u$.

\tweeplaatjesbreed {chain-care} {chain-graph}
{(a) A node $u$ with $|Z_u| = 5$, 
     and the relevant part of the quadtree. 
 (b) The graph $\Gamma$. 
 Tree edges are black (going right). To avoid clutter, we just show
 two wspd edges (green, going left).
}

Define a directed graph $\Gamma$ as follows:
we create two copies
of each vertex $u$ in $T$, called $\texttt{start}(u)$ and
$\texttt{end}(u)$, and we add a directed edge from
$\texttt{start}(u)$ to $\texttt{end}(u)$ for each
such vertex.
Furthermore,
we replace every edge $uv$ of $T$ ($u$ being the parent of $v$)
by two edges: one from $\texttt{start}(u)$ to $\texttt{start}(v)$,
and one from $\texttt{end}(v)$ to  $\texttt{end}(u)$. We
call these edges the \emph{tree}-edges.
Finally, for every pair $\{u,v\} \in \wspd(T)$, where
$S_v$ is wholly contained in the extended rightward cone
$\mathcal{C}_\rightarrow^{+}(c_u)$, we create a directed
edge from $\texttt{end}(u)$ to $\texttt{start}(v)$. These
edges are called \emph{wspd}-edges.
Figure~\ref {fig:chain-care+chain-graph} shows a small example.

\begin{claim}\label{clm:acyclic}
The graph $\Gamma$ is acyclic.
\end {claim}

\begin{proof}
Suppose $C$ is a cycle in $\Gamma$.
The tree-edges form an acyclic subgraph, so $C$
has at least one wspd-edge. Let $e_1, e_2, \ldots, e_z$ be
the sequence of wspd-edges along $C$,
and let $v_1, \ldots, v_z$ be
such that the endpoint of $e_i$
is of the form $\texttt{start}(v_i)$.
Finally, write
$C = e_1 \rightarrow C_1 \rightarrow e_2 \rightarrow C_2 \rightarrow \cdots
\rightarrow e_z \rightarrow C_z$, where $C_i$ is the sequence of
tree-edges between two consecutive wspd-edges.
Each $C_i$ consists of a (possibly empty) sequence of
$\texttt{start}-\texttt{start}$
edges, followed by one $\texttt{start}-\texttt{end}$
edge and a (possibly empty) sequence of $\texttt{end}-\texttt{end}$ edges.
Thus, the origin of the next wspd-edge $e_{i+1}$ is
an $\texttt{end}$-node for an ancestor or a
descendant of $v_i$ in $T$. In either case, by the definition of wspd-edges,
it follows that the leftmost point of $S_{v_{i+1}}$ lies
strictly to the right of the leftmost point of $S_{v_i}$. 
Indeed, write $e_{i+1} = (u_{i+1}, v_{i+1})$. Then $S_{v_{i+1}}$
lies strictly to the right of $S_{u_{i+1}}$, because 
$S_{v_{i+1}} \subseteq \mathcal{C}_\rightarrow^{+}(c_{u_{i+1}})$
and because $\{u_{i+1}, v_{i+1}\}$ is well-separated.
If $u_{i+1}$ is a descendant of $v_i$, then 
$S_{u_{i+1}} \subseteq S_{v_i}$ and the leftmost
point of $S_{u_{i+1}}$ cannot lie to the left of the leftmost
point of $S_{v_i}$, which implies the claim. 
If $u_{i+1}$ is an ancestor of $v_i$, then all of $S_{v_{i+1}}$ is
strictly to the right of $S_{v_i}$, and the claim follows again.
Thus,
the leftmost point of $S_{v_{i+1}}$ lies strictly to the right of the leftmost
point of $S_{v_{i}}$ and the leftmost point of
$S_{v_1}$ lies strictly to the right of the leftmost point in
$S_{v_z}$, which is absurd.
\end {proof}

Let $\leq_{\Gamma}$ be a topological ordering
of the nodes of $\Gamma$.

\begin{claim} \label {clm:order}
Any pair $(p, q)$ of points in $Z_u$ with $p \leq_\Gamma q$ satisfies
$q \notin \mathcal{C}_\leftarrow(p)$.
\end {claim}

\begin {proof}\label{clm:respecting_order}
Suppose for the sake of contradiction that $q \in \mathcal{C}_\leftarrow(p)$.
Let $v$, $w$ be the descendants of $u$ such that
$q \in P_v$, $p \in P_w$, 
and $\{v,w\} \in \wspd(T)$.
By Observation~\ref{obs:small-Phi}, $S_w$ lies completely
in the extended rightward cone $\mathcal{C}_\rightarrow^+(c_v)$, so
$\Gamma$ has an edge from $\texttt{end}(v)$ to $\texttt{start}(w)$.
Now the tree edges in $\Gamma$
require that the leaf with $q$
comes before $\texttt{end}(v)$ and the leaf with $p$ comes after
$\texttt{start}(w)$, and the claim follows.
\end{proof}

\tweeplaatjes {topo-order} {topo-chain}
{(a) A set of points, and all edges with a slope in $[-1,1]$.
     By Claim~\ref {clm:order}, these edges are all (possibly implicitly)
     present in $\Gamma$.
 (b) A possible ordering $\leq_\Gamma$ of the points that respects $\Gamma$.
}

Since all edges on $\UC(Z_u)$ have slopes in $[-1, 1]$, we immediately
have the following corollary.

\begin{cor}
The ordering $\leq_\Gamma$ respects the orders of $\UC(Z_u)$ and $\LC(Z_u)$.

\end{cor}

For every node $u \in T$, let $\leq_u$ be the order 
that $\leq_{\Gamma}$ induces on the leaf nodes corresponding to $Z_u$. 

\begin{claim}\label{clm:topo_sort}
All the orderings $\leq_u$ can be found in total time $O(n)$.
\end{claim}

\begin {proof}
To find the orderings $\leq_u$, perform a topological sort on
$\Gamma$, in linear time\footnote{Note that
$\Gamma$ has $O(n)$ edges, as
$|\wspd(T)| = O(n)$.}~\cite[Chapter~22.4]{CormenLeRiSt09}.
With each node $u$ of $T$ store a list $L_u$, initially empty.
We scan the nodes of $\Gamma$ in order. Whenever we see 
a leaf for a point $p \in P$, we append $p$ to the
at most $2k$ 
lists $L_u$ for the nodes $u$ with $p  \in Z_u$.
The total running time is 
$O(n + \sum_{u \in T} |Z_u|) = O(n)$, and $L_u$
is sorted according to $\leq_u$ for each $u \in T$.
\end{proof}

\begin{claim}\label{clm:findUCLC}
For any node $u \in T$, if $Z_u$ is sorted according to $\leq_{u}$,
we can find $\UC(Z_u)$ and
$\LC(Z_u)$ in
time $O(|Z_u|)$.
\end{claim}

\begin{proof}
We can find $\UC(Z_u)$ by a Graham-type pass through $L_u$.
An example of such a list is shown in Figure~\ref {fig:topo-chain}.
That is, we scan $L_u$ from left to right, 
maintaining
a tentative upper chain $U$, stored as a stack. 
Let $r$ be the rightmost point of $U$.
On scanning a new point $p$, we distinguish cases depending
in which of the four quadrants
$\mathcal{C}_\leftarrow(r)$, $\mathcal{C}_\uparrow(r)$, 
$\mathcal{C}_\rightarrow(r)$, or $\mathcal{C}_\downarrow(r)$
it lies in.
By Claim~\ref {clm:respecting_order}, we know that 
$p \notin \mathcal{C}_\leftarrow(r)$. 
If $p \in \mathcal{C}_\downarrow(r)$, we discard $p$ and continue to
the next point in $L_u$. 
If $p \in \mathcal{C}_\uparrow(r)$, we
pop $r$ from $U$ and reassess $p$ from the point of view of the
new rightmost point of $U$.
If $p \in \mathcal{C}_\rightarrow(r)$, we push $p$ onto $U$.

The algorithm takes $O(|Z_u|)$ time, because every point is pushed
or popped from the stack at most once and because it takes constant
time to decide which point to push or pop.
Now we argue correctness. For this, we use induction in order to
prove that after $i$ steps, we have correctly computed the
upper chain for the first $i$ points in $L_u$, $\UC(L_i)$. This clearly
holds for the first point. Now consider the cases for the $(i+1)$-th
point $p$. 
\begin{itemize}
\item
If $p \in \mathcal{C}_\downarrow(r)$, then $p$ is certainly
not on the upper chain. Furthermore, 
$\mathcal{C}_\downarrow(p) \subseteq \mathcal{C}_\downarrow(r)$,
so $p$ cannot conflict with any other point on $\UC(L_i)$,
so in this case $\UC(L_{i+1}) = \UC(L_i)$.
\item
If $p \in \mathcal{C}_\uparrow(r)$, then 
$\mathcal{C}_\uparrow(p) \subseteq \mathcal{C}_\uparrow(r)$ and $p$
must be on $\UC(L_{i+1})$. Furthermore, every point that we remove
from $\UC(L_i)$ has $p$ in its upper cone and cannot be on
$\UC(L_{i+1})$. Now let $r'$ be the first point of $\UC(L_i)$ that
is not popped. Since 
$\mathcal{C}_\leftarrow(r') \subseteq \mathcal{C}_\leftarrow(p)$ and since
the remainder of $\UC(L_i)$ lies inside of 
$\mathcal{C}_\leftarrow(r')$, there are no conflicts between $p$ and the
points we have not popped. Thus $\UC(L_{i+1})$ is computed correctly.
\item If $p \in \mathcal{C}_\rightarrow(r)$, then 
$\mathcal{C}_\uparrow(p) \subseteq \mathcal{C}_\uparrow(r) \cup
\mathcal{C}_\rightarrow(r)$, and $p$ is on $\UC(L_{i+1})$, because
$\mathcal{C}_\rightarrow(r)$ contains no points from $L_i$.
Futhermore, $\UC(L_i)$ is contained in $\mathcal{C}_\leftarrow(p)$,
so $p$ conflicts with no point on $\UC(L_i)$ and the result is correct.
\end{itemize}
This finished the inductive step and the correctness proof.
The lower chain is computed in an analogous manner.
\end{proof}

\begin{claim}\label{clm:find_closest}
For any node $u \in T$ and any pair $\{u,v\}$ in $\mathcal{P}_u$, given
$\UC(Z_u)$ and $\LC(Z_v)$, we can find the closest pair in $Z_u \otimes Z_v$
in time $O(|Z_u| + |Z_v|)$.
\end{claim}

\begin{proof}
Connect the endpoints of $\UC(Z_u)$ and $\LC(Z_v)$ to obtain a simple
polygon (note that the two new edges cannot intersect the chains,
because $\{u,v\}$ has direction $\phi=\pi/2$,
so by Observation~\ref {obs:small-Phi}
$\Phi_{uv} \subseteq [\pi/2-8\frac12\eps,\pi/2+8\frac12\eps]$
and all edges of the chains have
slopes in $[-1,1]$). Then use the algorithm of Chin and
Wang~\cite{ChinWa99} to find the constrained DT of
the polygon in time $O(|Z_u| + |Z_v|)$. The closest pair will appear
as an edge in this DT, and hence can be found in
the claimed time.\footnote{Actually, the resulting polygon is $x$-monotone,
so the most difficult part of the algorithm by Chin and
Wang~\cite{ChinWa99}, finding the visibility map of the
polygon~\cite{Chazelle91}, becomes much easier~\cite{GareyJoPrTa78}.
The problem may allow a much more direct solution, but since
we will later require Chin and Wang's algorithm in full
generality, we do not pursue this direction.}
\end{proof}

\begin{lem}\label{lem:findNN}
In total linear time, we can find for every $u \in T$ and for
every pair $\{u,v\} \in \mathcal{P}_u$ the closest pair in $Z_u \otimes Z_v$.
\end{lem}

\begin{proof}
By Claims~\ref{clm:topo_sort}, \ref{clm:findUCLC}, \ref{clm:find_closest},
the time to find all the closest pairs is  proportional to
\[
O(n + \sum_{u \in T} \sum_{\{u,v\} \in \mathcal{P}_u} (|Z_u| + |Z_v|))
= O(n + \sum_{u \in T} |Z_u|) = O(n),
\]
because every $v$ appears in only a constant number of $\mathcal{P}_u$'s.
\end{proof}

\paragraph{Putting it together.}
We thus obtain the main result of this section.

\begin{theorem}\label{thm:qtree->h}
Given a compressed quadtree $T$ for $P$ and \emph{$\wspd(T)$}, we can 
find a graph
$H$ with $O(n)$ edges such that $H$ contains all edges of $\emst(P)$.
It takes $O(n)$ time to construct $H$.
\end{theorem}

\begin{proof}
The fact that $H$ contains the EMST follows
from Lemmas~\ref{lem:constant-levels}, \ref{lem:constant-neighbors} and
\ref{lem:emst_nn}. The running time follows from the discussion at the
beginning of Steps 1 and 2 
and from Lemma~\ref{lem:findNN}.
\end{proof}

\subsection {Extracting the EMST}
  We want to extract $\emst(P)$, but no general-purpose deterministic linear 
  time pointer machine algorithm for this problem is known:
  the fastest such algorithm whose running time can be analyzed
  needs $O (n \alpha (n))$ steps~\cite {Chazelle00}.
  However, the special structure of the graph $H$ and the $c$-cluster quadtree
  $T$ make it possible to achieve linear time.

  We know that $H$  contains all EMST edges.
  Furthermore, by construction each edge of $H$ corresponds
  to a wspd-pair. Thus, we can associate each edge $e$
  of $H$ with two nodes $u$ and $v$ such that $\{u,v\}$ is
  the wspd-pair for the endpoints of $e$. The pruning operation in
  Step~2 of Section~\ref{sec:emstsup} ensures that each
  node is associated with $O(1)$ edges of $H$, and we store
  a list of these edges at each node of $T$.
  Now we use Theorem~\ref{thm:cluster-compressed-equiv}
  to convert our quadtree into a $c$-cluster quadtree $T$.
  During this conversion, we can preserve the information
  about which edges of $H$ are associated with which 
  nodes of $T$, because each old square overlaps with 
  only a constant number of new squares of similar
  size. A special case are those edges that have an
  endpoint associated with a compressed child.
  During the conversion of Theorem~\ref{thm:cluster-compressed-equiv},
  compressed children either become regular squares (during the balancing
  operation), or they correspond to $c$-clusters and are replaced
  by representative points in the parent tree. In the former case,
  we handle the compressed child just like any regular square, in the latter
  case, we associate $e$ with the square that contains the representative
  point for the $c$-cluster.
  
  Next, we would like ensure for each edge $e$ of $H$ that the associated
  squares in $T$ have size between $\eps|e|/2$ and $2\eps|e|$, where
  $|e|$ denotes the length of $e$. 
  For the endpoints that were associated with regular squares in the
  original quadtree, such a square can be found by considering a constant 
  number of ancestors and descendants in $T$, by 
  Claim~\ref{clm:quad-squares}.  If the associated square was a compressed
  child that has become a regular square, we may need to consider more than
  a constant number of ancestors, but each such ancestor is considered only
  a constant number of times, since the compressed child has a 
  constant number of associated edges. If $e$ has an endpoint 
  that is now associated with a representative point, we may need to subdivide
  the square containing the representative point, but by 
  Corollary~\ref{cor:c-cluster-QT-extended} the total work is linear.
  Thus, in total linear time we can obtain a $c$-cluster
  tree $T$ such that each square of $T$ is associated with $O(1)$
  edges of $H$ and such that the two associated square of each 
  edge $e$ of $H$  contain the endpoints of $e$ and have size
  $\Theta(\eps|e|)$. 

  By the cut property of minimum spanning trees,
  $\emst(P)$ is connected within each $c$-cluster. Thus, we can process
  the clusters bottom-up, and we
  only need to find the EMST within a $c$-cluster given
  that the points in each child are already connected. 
  Within this cluster, $T$ is a regular uncompressed quadtree,
  and we can use the structure of $T$ to perform an appropriate
  variant of \boruvka's MST algorithm~\cite{Boruvka26,Tarjan83} in linear time.

  \begin{lem}\label{lem:extract-emst}
    Let $T'$ be a subtree of $T$ corresponding to a $c$-cluster, and
    let $E$ be the edges in $H$ associated with $T'$.
    Then $\emst(P) \cap E$ can be computed in time $O(|E| + |V(T')|)$.
  \end{lem}

\begin{proof}
  Let $\ell$ be the size of the root square of $T'$.
  Through a level order traversal of $T'$ we group the
  squares in $V(T')$ by height into layers $V_1$, $V_2$, $\ldots$, $V_h$
  (where $V_1$ is the bottommost layer, and $V_h$ contains
  only the root). The squares in $V_i$ have size $\ell/2^{h-i}$.
  As stated above, each square $S$ has a constant number of
  associated edges in $E$ that have one endpoint in $S$ and length
  length  between $|S|/2\eps$ and $2|S|/\eps$.
  To find the EMST, we subdivide the edges into
  sets $E_i$, where $E_i$ contains all edges with length
  in $[\ell/(\eps2^{h-i}), \ell/(\eps2^{h-i-1}))$. Given the
  $V_i$, we can determine the sets $E_i$ in total time
  $O(|E| + |V(T')|)$, as the edges for $E_i$ are
  associated only with squares in $V_{i-\alpha}$, $V_{i-\alpha+1}$,
  $\ldots$, $V_{i+\alpha}$,
  for some constant $\alpha$. 
  Note that every edge
  in $E_i$ is crossed by $O(1)$ other edges in $E_i$, because all
  $e \in E_i$ have roughly the same length and because every pair of squares in
  $V_i$ has only a constant number of associated edges in $E_i$.

  Now we compute the EMST by processing the sets $E_1$,
  $\ldots$, $E_h$ in order. Here is how to
  process $E_i$. We consider the squares in $V_i$.
  Assume that we know for each square of $V_i$ the connected
  component in the current partial EMST it meets (initially
  each $c$-cluster is its own component). By the cut property, every
  square $S$ meets
  only one connected component, as $S$ is much smaller than the
  edges in $E_i$.
  Eliminate all edges in $E_i$ between squares in the same component,
  and 
  remove duplicate edges between
  each two components, keeping only the shortest of these edges
  (this takes $O(|E_i|)$ time with appropriate pointer manipulation).
  Then find the shortest edge out of each component
  and add these edges to the partial EMST. Determine the new components
  and merge their associated edge sets.
  This sequence of steps is called
  a \emph{\boruvka-phase}.
  Perform \boruvka-phases until $E_i$ has no edges left.

  By the crossing-number inequality~\cite[Theorem~4.3.1]{Matousek02}, the 
  number of edges considered in
  each phase is proportional to the number $r$ of components with an outgoing
  edge in that phase. Indeed, viewing each component as a supervertex,
  we have an embedding of a graph with $r$ vertices and $z$ edges such
  that there are $O(z)$ crossings (since every edge $e \in E_i$ is crossed
  by $O(1)$ other edges in $E_i$). Thus, the crossing number
  inequality yields $z^3/r^2 \leq \beta z$, 
  for some
  constant $\beta > 0$, so $z = O(r)$.  Since the number of
  components at least halves in each phase, and since initially there
  are at most $|V_i|$ components, the total time
  for $E_i$ is $O(|E_i| + |V_i|)$. Finally, label each
  square in $V_{i+1}$ with the component  it meets and proceed
  with round $i+1$.
  In total, processing $T$ takes time $O(|V(T')| + |E|)$, as desired.
\end{proof}

\subsection{Finishing Up}
We conclude:
\begin {theorem} \label {thm:wspd->dt}
Let $P$ be a planar point set and $T$ be a compressed quadtree 
or a $c$-cluster quadtree 
for $P$.  Then $\DT(P)$ can be computed in time $O(|P|)$.
\end {theorem}

\begin{proof}
If $T$ is a $c$-cluster quadtree, invoke 
Theorem~\ref{thm:cluster-compressed-equiv} to convert it
to a compressed quadtree.
Then use Theorem~\ref{thm:wspd} to obtain $\wspd(T)$.
Next, apply Theorem~\ref{thm:qtree->h} to compute the
supergraph $H$ of $\emst(P)$. After that, if necessary, convert $T$ to a 
$c$-cluster quadtree for $P$ via Theorem~\ref{thm:cluster-compressed-equiv},
and apply Lemma~\ref{lem:extract-emst} to each $c$-cluster, in a bottom-up
manner, to extract $\emst(P)$. Finally, apply the algorithm
by Chin and Wang~\cite{ChinWa99} to find $\DT(P)$.
All this takes time $O(|P|)$, as claimed.
\end{proof}

\section {From Delaunay Triangulations to $c$-Cluster Quadtrees}
\label {sec:dt->qt}

    For the second direction of our equivalence we need to show
    how to compute a $c$-cluster quadtree for $P$ when given
    $\DT(P)$. This was already done
    by Krznaric and Levcopolous~\cite{KrznaricLe95,KrznaricLe98},
    but their algorithm works in a stronger model of
    computation which includes the floor function and allows access to data
    at the bit level.
    As argued in the introduction, we prefer the real RAM/pointer machine,
    so we need to do some work to adapt their
    algorithm to our computational model.
    In this section we describe how Krznaric and Levcopolous's
    algorithm can be modified to avoid bucketing and bit-twiddling techniques.
    The only difference is that in the resulting $c$-cluster
    quadtree the squares for the $c$-clusters are not perfectly
    aligned with the squares of the parent quadtree. In our setting,
    this does not matter. The goal of this section is to prove the following
    theorem.

\begin{theorem}\label{thm:dt->c-cluster-qt}
Given $\DT(P)$, we can compute a $c$-cluster quadtree
for $P$ in linear deterministic time on a pointer machine.
\end{theorem}

In the following, we will refer to the paper by Krznaric and
Levcopolous~\cite{KrznaricLe98} as KL.
Our description is meant to be self-contained; however, we refer the reader to KL for more intuition and a more elaborate description of the main ideas.

\subsection{Terminology}
We begin by recalling some terminology from KL.

\begin{itemize}
\item \textbf{neighborhood.}
The \emph{neighborhood} of a square $S$ of a quadtree
consists of the 25 squares of size $|S|$ concentric around
$S$ (including $S$); see Figure~\ref{fig:neighbors}.
\item \textbf{direct neighborhood.}
The \emph{direct neighborhood} of a square
$S$ consists of the 9 squares of size $|S|$ directly adjacent to $S$
(including $S$); see Figure~\ref{fig:neighbors}.
\item \textbf{star of a square.}
Let $P$ be a planar point set, and let $S$ be a square. The \emph{star}
of $S$, denoted by $\bigstar(S)$, is the set of all edges $e$ in $\DT(P)$
such that (i) $e$ has one endpoint inside $S$ and one endpoint outside the
neighborhood of $S$; and (ii) $|e| \leq 16|S|$, where $|e|$ is
the length of $e$.
\item \textbf{dilation.}
Let $P$ be a planar point set, and $G$ a connected plane graph with
vertex set $P$. The \emph{dilation} of $P$ is the distortion between
the shortest path metric in $G$ and the Euclidean distance, i.e.,
the maximum ratio, over all pairs of distinct points $p,q \in P$, between
the length of the shortest path in $G$ from $p$ to $q$, and $|pq|$. There
are many families of planar graphs whose dilation is bounded by a
constant~\cite{DasJo89}.
In particular, for any planar point set $P$, the
dilation of $\DT(P)$ is bounded
by $2\pi/(3\cos(\pi/6)) \leq 2.42$~\cite{KeilGu92}.
\item \textbf{orientation.}
The \emph{orientation} of a line segment $e$ is the angle
the line through $e$ makes with the $x$-axis.
\end{itemize}

\eenplaatje {neighbors} {The neighborhood of a square $S$. The direct neighbors
are shown in dark blue, the others in light blue.}

\subsection{Preprocessing}
\label{sec:dt->c-cluster-preprocess}

By Theorem~\ref{thm:c-cluster-tree}, we can obtain a $c$-cluster
tree $T_c$ for $P$ in linear time, given $\DT(P)$. Thus, we only
need to construct the regular quadtrees $T_u^Q$ for each node 
$u$ in $T_c$. This is done by processing each node of $T_c$
individually. First, however, we need to perform a preprocessing step
in order to find  for each edge $e$ of $\DT(P)$ the node of  $T_c$
that is the least common ancestor of $e$'s endpoints.
For every node $u \in T_c$, we define $\texttt{out}(u)$ as the
set of edges in $\DT(P)$ that have exactly one endpoint in
$P_u$ and both endpoints in $P_{\parent u}$. Clearly, every edge is
contained in exactly two sets $\texttt{out}(u)$ and $\texttt{out}(v)$,
where $u$ and $v$ are siblings in $T_c$.
The following is a simple variant of a lemma from
KL~\cite[Lemma~3]{KrznaricLe98}.

\begin{lem}[Krznaric-Levcopolous]\label{lem:out}
Let $P$ be a planar $n$-point set.
Given $DT(P)$ and a $c$-cluster tree $T_c$ for $P$, the sets
\emph{$\texttt{out}(u)$} for every node $u \in T_c$ can be found in overall
$O(n)$ time and space on a pointer machine.
\end{lem}

\begin{proof}
KL show how to reduce the problem of determining the sets
$\texttt{out}(u)$ to $O(n)$ off-line least-common ancestor (lca) queries 
in two appropriate trees. For the lca-queries,
they invoke an algorithm by Harel and Tarjan~\cite{HarelTa84} that requires
the word RAM. However, since all lca-queries are known in
advance (i.e., the queries are \emph{off-line}), we  
may instead use an algorithm by
Buchsbaum \etal~\cite[Theorem~6.1]{BuchsbaumGeKaRoTaWe08} which 
requires $O(n)$ time and space on a pointer machine.
\end{proof}

\subsection{Processing a Single Node of $T_c$}
\label{sec:one-node}
We now describe the preprocessing that is necessary on a single
node $u$ of $T_c$ before the quadtree $T_u^Q$ can be constructed.
Let $v_1, v_2, \ldots, v_m$
be the children of $u$. For each child $v_i$, let
$\delta_i \eqdef d(P_{v_i}, P_u \setminus P_{v_i})$.

\begin{claim}\label{clm:delta_i}
For $i = 1, \ldots, m$, \emph{$\texttt{out}(v_i)$} contains an edge of
length $\delta_i$.
\end{claim}

\begin{proof}
If $\DT(P)$ contains an edge $e$ with an endpoint in $P_{v_i}$ and
with length $\delta_i$, then $e$ must
be in $\texttt{out}(v_i)$, by the definition of a $c$-cluster.
Since $\emst(P)$ is a subgraph of $\DT(P)$, it thus suffices to show that
$\emst(P)$ contains such an edge.
Consider running Kruskal's MST algorithm on $P$. According to the
definition of a $c$-cluster, by the time the algorithm considers the edge $e$
that achieves $\delta_i$, the partially constructed EMST contains exactly
one connected component that has precisely the points in $P_{v_i}$.
Therefore, $e \in \emst(P)$, and the claim follows.
\end{proof}

\paragraph{Initialization.}
By scanning the sets $\texttt{out}(v_i)$, we determine a child $v_j$
with minimum $\delta_j$ (by Claim~\ref{clm:delta_i} a shortest edge in
$\texttt{out}(v_i)$ has length $\delta_i$).
We may assume that $j = 1$. Let $S_1$ be a square that contains
$P_{v_1}$ and that has side-length $\delta_1 / 8$. Let $\alpha$ be
the smallest integer such
that four squares of size $2^{\alpha-1}\delta_1/8$ cover all of $P_u$.
Lemma~\ref {lem:c-cluster-QT} implies that $\alpha = O(m)$.

The goal is to compute $T_u^Q$, the balanced regular quadtree aligned at 
$S_1$ such that each $P_{v_i}$ is contained in squares of size $\delta_i/8$.
To begin, we use $S_1$ to initialize $T_u^Q$ as the partial 
balanced quadtree $T_u^Q$ shown in Figure~\ref{fig:initialqt}.
\eenplaatje [scale=0.8] {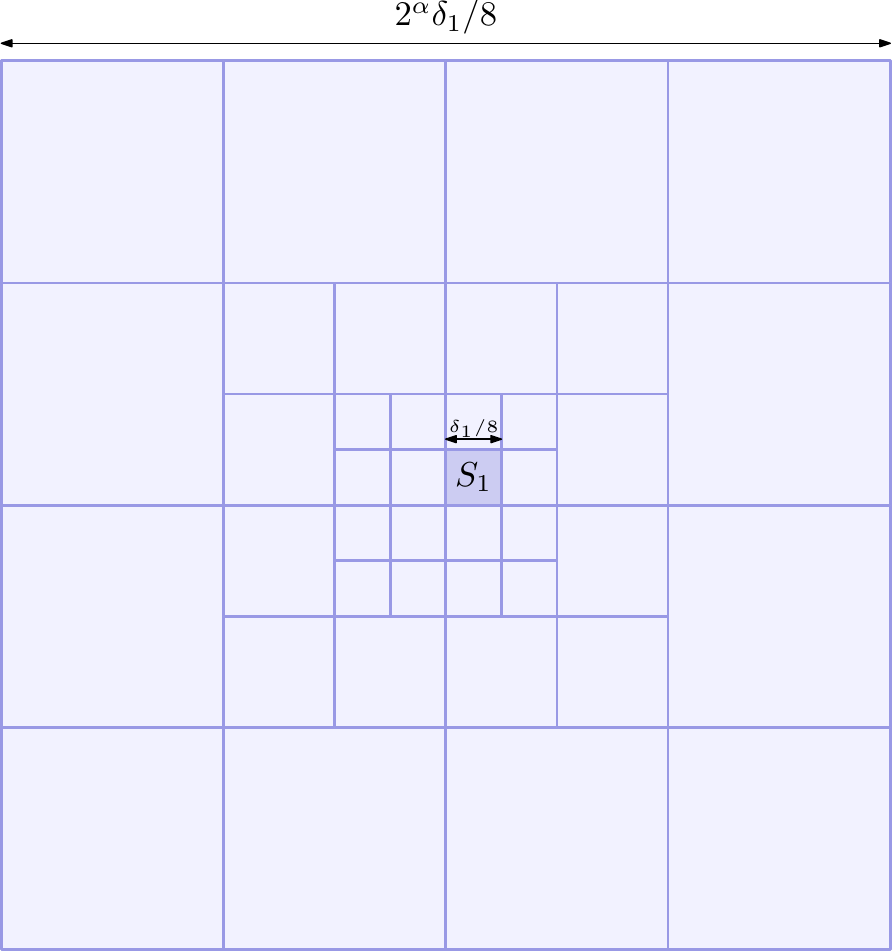} {The initial quadtree.}
Every square $S$ of $T_u^Q$ stores the following fields:
\begin{itemize}
\item $\texttt{parent}$: a pointer to the parent square, $\texttt{nil}$ for
       the root;
\item $\texttt{children}$: pointers for the four children of $S$,
     $\texttt{nil}$ for a leaf;
\item $\texttt{neighbors}$: links to the four orthogonal neighbors of $S$
      in the quadtree $T_u^Q$ with size $|S|$
      (or size $2|S|$, if no smaller neighbor exists);
\end{itemize}

The fields $\texttt{parent}$, $\texttt{children}$, and $\texttt{neighbors}$
are initialized for all the nodes in $T_Q$.

\begin{lem}\label{lem:node-init}
The total time for the initialization phase is
\emph{$O(m + \sum_{i=1}^{m} |\texttt{out}(v_i)|)$}.
\end{lem}

\begin{proof}
By Lemma~\ref{lem:c-cluster-QT}, the initial size of $T_u^Q$ is $O(m)$.
All other operations consist of scanning the \texttt{out}-lists
or are linear in the size of $T_u^Q$.
\end{proof}

\subsection {Building the Tree $T_u^Q$}

\begin{algorithm}
  $\texttt{explore}(\mathcal{S}, \texttt{maxsize})$
  \begin{enumerate}
    \item\label{step:activeS} Set $\texttt{active} := \mathcal{S}$.
    \item Set $\texttt{newActive} := \emptyset$.
    \item\label{step:loop}
          Until the squares in $\texttt{active}$ have size greater than
          $\texttt{maxsize}$:
    \begin{enumerate}
      \item\label{step:activeLoop} For every square $S$ in 
            $\texttt{active}$ call the function $\texttt{findStar}(S)$ 
      to determine $\bigstar(S)$.  Append $\parent S$ to 
      $\texttt{newActive}$, if it is not present yet.
      \item\label{step:starLoop} For every
            edge $e \in \bigcup_{S \in \texttt{active}} \bigstar(S)$,
            if $e$ has an endpoint in an undiscovered cluster,
      call the function $\texttt{newCluster}(S,e)$, and append
      all the squares returned by this call to $\texttt{newActive}$.
      \item \label{step:newActive} Set $\texttt{active} := \texttt{newActive}$.
    \end{enumerate}
  \end{enumerate}
 $\texttt{newCluster}(S,e)$
  \begin{enumerate}
    \item   Walk along $e$  through the current $T_u^Q$ to find the
            square $S'$ of $T_u^Q$ that contains the other endpoint
            of $e$.
            This tracing is done by following the appropriate
            $\texttt{neighbor}$ pointers from $S$.
     \item \label{step:refine} Refine $T_u^Q$ for the new cluster, and
            let $\mathcal{S}'$ be the set of leaf squares containing the newly
            discovered cluster.
      \item Call 
            $\texttt{explore}(\mathcal{S}', 
      \text{size of squares in \texttt{active}})$.
      Afterwards, return the $\texttt{active}$ squares from the recursive call.
  \end{enumerate}
  \caption{Computing a $c$-cluster quadtree for the children of
    a $c$-cluster.}
  \label{alg:explore}
\end{algorithm}

Now we build the tree $T_u^Q$ by a traversing $\DT(P)$ in
a way reminiscent of Dijkstra's algorithm~\cite{CormenLeRiSt09}.
In their algorithm, KL make extensive use of the floor function 
in order to locate points inside their quadtree squares. The 
purpose of this section is to argue that this point location work 
can be done through local traversal of the quadtree, without 
the floor function.
Refer to Algorithm~\ref{alg:explore}.
The heart of the algorithm is the procedure $\texttt{explore}$, 
which is initially called as
$\texttt{explore}(\{S_1\},2^{\alpha-1}\delta_1/8)$.
The procedure $\texttt{explore}$ builds the tree $T_u^Q$ level
by level, beginning with the level of $S_1$.
At each point, it maintains
a set $\texttt{active}$ of all squares at the current level that contain 
a cluster that has already been processed. For each such square $S$,
it calls a function $\texttt{findStar}$. This function
returns all edges of the Delaunay triangulation that have one endpoint in
$S$ and have length $\alpha|S|$, for a constant $\alpha$. Using 
$\texttt{findStar}$
we can  new clusters whose distance from the active clusters
is comparable to the size of the squares in the current level.
We will say more about the implementation $\texttt{findStar}$ below.
For each new cluster, we call the procedure $\texttt{newCluster}$ 
which adds more squares to $T_u^Q$ to accommodate the new cluster and 
recursively explores the short edges out of this new cluster. After
the recursive call has finished, we can continue the exploration of the
tree at the current level. 

We now give the details
for the refinement in Step~\ref{step:refine} of $\texttt{newCluster}$:
Let $v_j$ be
the cluster that contains the other endpoint $q$ of $e$ (we can find $v_j$
in constant time, since $e \in \texttt{out}(v_j)$, and
since for each edge we store the two clusters whose $\texttt{out}$-lists
contain it). Subdivide the current leaf square containing $q$ (and 
possibly also its neighbors if they contain points from $P_{v_j}$) 
in quadtree-fashion until $P_{v_j}$ is contained in squares of 
size $\delta_j/8$.  Then balance the quadtree and update the 
$\texttt{neighbor}$ pointers
accordingly. 

The algorithm is recursive, and at each point there exists
a sequence $\mathcal{E}_1$, $\mathcal{E}_2$, $\ldots$, $\mathcal{E}_z$
of instantiations (i.e., stack frames) 
of $\texttt{explore}$, where $\mathcal{E}_{i+1}$
was invoked by $\mathcal{E}_{i}$. Each $\mathcal{E}_i$ has a set
$\texttt{active}_i$ of active
squares, such that all squares in each $\texttt{active}_i$
have the same size, and such that the squares in $\texttt{active}_{i+1}$
are not larger than the squares in $\texttt{active}_i$.
We say that a square is active if it is contained in
$\texttt{active}^T \eqdef \bigcup_i \texttt{active}_i$.
The neighborhood of $\texttt{active}^T$
is the union of the neighborhoods of all boxes in $\texttt{active}$.
We maintain the following invariant:

\begin{invariant}\label{inv:neighborhood}
At all times during the execution of \emph{$\texttt{explore}$},
all undiscovered $c$-clusters lie outside the neighborhood of
\emph{$\texttt{active}^T$}.
\end{invariant}

\begin{claim}\label{clm:invariant}
Invariant~\ref{inv:neighborhood} is maintained by \emph{$\texttt{explore}$}.
\end{claim}

\begin{proof}
The set $\texttt{active}^T$ only changes in Steps~\ref{step:activeS} and
\ref{step:newActive}. The invariant is maintained in
Step~\ref{step:activeS}, since the size of the squares in $\mathcal{S}$ 
(i.e., $\delta_i/8)$ is chosen such that their neighborhoods 
can contain no point from any other cluster. 

Let us now consider Step~\ref{step:newActive}. The
set $\texttt{newActive}$ contains two kinds of squares:
(i) the parents of squares processed in the current iteration
of the main loop; and (ii) squares that were added to
$\texttt{newActive}$ after a recursive call. We only need
to focus on squares of type (i), since squares of type (ii)
are already added to $\texttt{active}^T$ during the recursive
call.
Suppose that $\texttt{active}^T$ contains a square $S$ whose neighborhood
has a point $p \in P$ in an undiscovered cluster. Since $S \in
\texttt{active}^T$,
there is a point $q \in P \cap \mathcal{S}$, and by the definition
of neighborhood, we have $d(p,q) \leq 3|S|$.
However, since the dilation of $\DT(P)$ is at most
$2.5$~\cite{KeilGu92}, $\DT(P)$ contains a path $\pi$ of length at most
$8|S|$ from $p$ to $q$. Let
$p'$ be the last discovered point along $\pi$. The point $p'$ lies
in an active square $S'$ with $|S'| \geq |S|$, and the edge $e$ leaving
$p'$ on $\pi$ has length at most $8|S'|$. Therefore,
$e \in \bigstar(S'')$  for a descendant $S''$ of $S'$, which
contradicts the fact that $p'$ is the last discovered point along
$\pi$.
\end{proof}

\begin{lem}\label{lem:node-explore}
The total running time of \emph{$\texttt{explore}$}, excluding the
calls to \emph{$\texttt{findStar}$}, is 
\emph{$O(m + \sum_{i=1}^{m} |\texttt{out}(v_i)|)$}.
\end{lem}

\begin{proof}
All squares appearing in $\texttt{active}^T$ are ancestors of non-empty
leaf squares in the final tree $T_u^Q$. Therefore, by
Lemma~\ref{lem:c-cluster-QT},
the total number of iterations for the loop in
Step~\ref{step:activeLoop} is $O(m)$. Furthermore, 
$\bigstar(S)$ contains only edges of length $\Theta(|S|)$, so every
edge appears in only a constant number of stars. It follows that the
total size of the $\bigstar$-lists, and hence the total number of
iterations of the loop in Step~\ref{step:starLoop} is
$O(\sum_{i=1}^{m} |\texttt{out}(v_i)|)$.

It remains to bound the time for tracing the edges and balancing the
tree.  Since $T_u^Q$ is balanced and since $\bigstar(S)$ contains
only edges of length $\Theta(|S|)$, the tracing along the $\texttt{neighbor}$
pointers of an edge takes constant time (since we traverse a constant number
of boxes of size $\Theta(|S|)$). By Invariant~\ref{inv:neighborhood},
the other endpoint of the edge  is contained in a leaf square
of the current $T_u^Q$ of size $\Theta(|S|)$. (This is because the quadtree
is balanced and because the other endpoint of the edge lies outside
the neighborhood of the active squares.) Therefore, the time to build
the balanced
quadtree for the new leaf squares containing the newly discovered cluster
can be charged to the corresponding nodes in the final $T_u^Q$, of which
there are $O(m)$. Furthermore, note that by
Invariant~\ref{inv:neighborhood}, balancing the quadtree for the newly
discovered leaf squares does not affect any descendants of the active
squares.
\end{proof}

\subsection {Implementing $\texttt{findStar}$}

KL show how to exploit the geometric properties of the Delaunay triangulation
in order to implement the function $\texttt{findStar}$,
quickly. For this, they store two additional fields with each
active square, called $\texttt{characteristic}$ and
$\texttt{shortcuts}$~\cite[Section~6]{KrznaricLe98}, and they explain how to
maintain these lists throughout the procedure. This part of the algorithm works 
on a real RAM/pointer machine without any further modification, so we 
just state their result.
\begin{lem}\label{lem:findStar}
The total time for all calls to \emph{$\texttt{findStar}$} and the
maintenance of the required data structures is
\emph{$O(m + \sum_{i=1}^{m} |\texttt{out}(v_i)|)$}.
\qed
\end{lem}

\subsection {Putting Everything Together}
We can now finally prove Theorem~\ref{thm:dt->c-cluster-qt}.

\begin{proof}[Proof of Theorem~\ref{thm:dt->c-cluster-qt}]
First, we use Theorem~\ref{thm:c-cluster-tree} to find a $c$-cluster
tree $T_c$ for $P$ in $O(n)$ time. Next, we use the algorithm from
Section~\ref{sec:dt->c-cluster-preprocess} to preprocess the tree.
By Lemma~\ref{lem:out}, this also takes $O(n)$ time. Finally,
we process each node of $T_c$ using the algorithm from 
Section~\ref{sec:one-node}. 
By Lemmas~\ref{lem:node-init}, \ref{lem:node-explore},
and \ref{lem:findStar}, this takes total time
$\sum_{j} 1+ |\texttt{out}(v_j)|$, where the sum ranges over
all the nodes of $T_c$. This sum is $O(n)$ because there are
$O(n)$ nodes in $T_c$, and because every edge of $\DT(P)$
appears in exactly two $\texttt{out}$-lists. Hence, the total running
time is linear, as claimed.
\end{proof}

\section {Applications}
\label {sec:applications}

  As mentioned in the introduction, 
  our result yields
  deterministic versions of several recent randomized algorithms
  related to DTs. 
  Firstly, we can immediately derandomize
  an algorithm for hereditary DTs by
  Chazelle~\etal~\cite{ChazelleDeHuMoSaTe02,ChazelleMu11}:

  \begin {cor}\label{cor:split}
     Let $P$ a planar $n$-point set, and let $S \subseteq P$.
     Given $\DT(P)$, we can find $\DT(S)$
     in deterministic time $O(n)$ on a pointer machine.
  \end {cor}

  \begin{proof}
     Use Theorem~\ref{thm:dt->c-cluster-qt}  to find a $c$-cluster
     quadtree $T$ for $P$,
     remove the leaves for $P \setminus S$ from $T$ and trim it
     appropriately.\footnote{Deleting $P \setminus S$ might
     create new $c$-clusters. However, since we are aiming
     for running time $O(n)$, we can  apply Theorem~\ref{thm:wspd->dt}
     to a partly compressed quadtree that may contain
     long paths where every node has only one child.}
     Finally, apply Theorem~\ref{thm:wspd->dt} to extract $\DT(S)$ from $T$,
     in time $O(n)$.
  \end{proof}

  Secondly, we obtain deterministic analogues of the algorithms by
  Buchin~\etal~\cite {BuchinLoMoMuXX} to preprocess imprecise point sets
  for faster DTs. For example, we can prove the following:
  \begin {cor}\label{cor:IDT}
    Let $\mathcal{R} = \langle R_1, R_2, \ldots, R_n\rangle$ be
    a sequence of 
    $n$ $\beta$-fat
    planar regions so that 
    no point in $\R^2$ 
    meets more than $k$ of them. 
    We can preprocess $\mathcal{R}$ in $O(n \log n)$
    deterministic time into an $O(n)$-size data structure
    so that given a sequence of  $n$ points 
    $P = \langle p_1, p_2, \ldots, p_n \rangle$ 
    with $p_i \in R_i$ for all $i$, we can find $\DT(P)$
    in deterministic time $O(n \log(k/\beta))$ on a pointer
    machine.
  \end {cor}

\begin{proof}
    The method of
    Buchin~\etal~\cite[Theorem~4.3 and Corollary~5.6]{BuchinLoMoMuXX}
    proceeds by computing a
    representative quadtree $T$ for $\mathcal{R}$.
    Given $P$, the algorithm finds for every point in $P$
    the leaf square of $T$ that contains it, and then uses this
    information to obtain a compressed quadtree $T'$ for $P$ in time
    $O(n \log(k/\beta))$. However, $T'$ is \emph{skewed}
    in the sense that not all its squares need to be perfectly
    aligned and that some squares can be cut off.
    However, the authors argue
    that even in this case $\wspd(T)$ takes $O(n)$ time and
    yields a linear-size WSPD~\cite[Appendix~B]{BuchinLoMoMuXX}.
    The main observation~\cite[Observation~B.1]{BuchinLoMoMuXX}
    is that any (truncated) square $S$ in $T'$ is adjacent to at least
    one square whose area is at least a constant fraction of the
    area $S$ would have without clipping.
    Since in skewed quadtrees the size of a node is at most half
    the size of its parent, the argument of Lemma~\ref{lem:constant-levels}
    still applies.
    To see that Lemma~\ref{lem:constant-neighbors} holds,
    we need to check that the volume argument goes through.
    For this, note that by the main observation of Buchin~\etal, we
    can assign every square $R_w$ (the notation is as in the proof of
    Lemma~\ref{lem:constant-neighbors}) to an adjacent square of
    comparable size at distance $O(\eps)$ from $A$. Since every such
    square is charged by disjoint descendants from a constant number of
    neighbors, the volume argument still applies, and
    Lemma~\ref{lem:constant-neighbors} holds. Lemma~\ref{lem:findNN}
    only relies on well-separation and the combinatorial structure of
    $T$, and hence remains valid.
    Finally, in order to apply Lemma~\ref{lem:extract-emst},
    we need to turn $T'$ into a $c$-cluster quadtree,
    which takes linear time by Theorem~\ref{thm:cluster-compressed-equiv}.
    Thus, the total running time is $O(n \log(k/\beta)$, as claimed.
\end{proof}

  Finally, Buchin and Mulzer~\cite{BuchinMu11} 
  showed that for word RAMs, DTs are no harder than sorting. 
  We can now do it deterministically.
  Let $\texttt{sort}(n)$ be the time to sort $n$ integers
  on a $w$-bit word RAM. The best deterministic bound
  for $\texttt{sort}(n)$ is 
  $O(n \log\log n)$~\cite{Han04}.\footnote{For specific ranges of $w$, 
  we can do better.
  For example, if $w = O(\log n)$,  radix sort shows
  that $\texttt{sort}(n) = O(n)$~\cite{CormenLeRiSt09}.}
  \begin {cor}\label{cor:wramDT}
     Let $P$ be a planar $n$-point set given by $w$-bit integers, for some 
     word-size $w \geq \log n$. 
     We can find $\DT(P)$ in deterministic time $O(\emph{\texttt{sort}}(n))$ on
     a word RAM supporting the 
     $\emph{\texttt{shuffle}}$-operation.\footnote{For two $w$-bit words,
     $x = x_1 \ldots x_w$ and $y = y_1 \ldots, y_w$, we define
     $\texttt{shuffle}(x,y)$ as the $2w$-bit word
     $z = x_1 y_1 x_2 y_2 \ldots x_wy_w$.}
  \end {cor}

\begin{proof}
      Buchin and Mulzer~\cite{BuchinMu11} show how to find a
      compressed quadtree $T$  for $P$ in time $O(\texttt{sort}(n))$,
      using the $\texttt{shuffle}$-operation. They actually do not
      find the squares of the quadtree, only the combinatorial
      structure of $T$ and the bounding boxes $B_v$.  It is
      easily seen that the algorithm $\wspd$ also works
      in this case.

      To apply Lemma~\ref{lem:constant-levels}, we need to check that
      the sizes of the bounding boxes decrease geometrically down the
      tree. For this, consider a node $v \in T$ with associated point
      set $P_v$ and the quadtree square $S_v$ (i.e., the smallest aligned
      square of size $2^{l}$ such that the coordinates of all points
      in $P_v$ share the first $w-l$ bits). Let $B_v$ be the bounding box
      of $P_v$, and let $l'$ be such that $2^{l'+1} \geq |B_v| \geq 2^{l'}$.
      Clearly, $B_v$ meets at most nine aligned squares of size $2^{l'}$,
      arranged in a $3 \times 3$ grid. Hence, any descendant $\child v$ of
      $v$ that is at least five levels below $v$ must have
      $|B_{\child v}| \leq |S_{\child v}| \leq |B_v|/2$, since after at most
      four (compressed) quadtree divisions the  squares for $B_v$ have
      been separated. Thus, the proof of Lemma~\ref{lem:constant-levels}
      goes through as before, if we choose $k$ larger and
      consider every fifth node along the chain $u_1, u_2, \ldots, u_k, u$.

      Lemma~\ref{lem:constant-neighbors} still holds, because
      every bounding box $B_v$ is contained in a (possibly much
      larger) square $S_v$, so the volume argument applies.
      Also, Lemma~\ref{lem:findNN} only relies on well-separatedness
      and the combinatorial structure of $T$, so we can find
      the graph $H$ in linear time. After that, it takes $O(n)$
      time to compute $\emst(P)$, using the transdichotomous minimum
      spanning tree algorithm by Fredman and Willard~\cite{FredmanWi94}.
\end{proof}

\section{Conclusions}

We strengthen the connections between proximity structures
in the plane and sharpen several known results between them.
Even though our results are optimal, the underlying algorithms are 
still quite subtle, and it may be of interest to see whether some of
them can be simplified. It is also interesting to see whether systematic
derandomization techniques, like $\eps$-nets, can be useful to yield 
alternative deterministic algorithms for some of the problems considered
here. Finally, some of the previous results also apply to higher dimensions,
whereas we focus exclusively on the plane. Can we obtain analogous 
derandomizations for $d \geq 3$?

\section*{Acknowledgments}

This work was initiated while the authors were visiting the Computational
Geometry Lab at Carleton University. We would like to thank the group
at Carleton and especially our hosts Jit Bose and Pat Morin for their
wonderful hospitality and for creating a great research environment.
We also would like to thank Kevin Buchin and Martin N\"ollenburg for
sharing their thoughts on this problem with us.
Work on this paper by M. L\"offler has been supported by the Office of
Naval Research under MURI grant N00014-08-1-1015.
Work by W. Mulzer has been supported in part by NSF grant CCF-0634958,
NSF CCF 083279, and a Wallace Memorial Fellowship in Engineering.

\small

\bibliographystyle {abbrv}
\newcommand{\SortNoop}[1]{}

\normalsize

\appendix
  \section {Computational Models}
  \label{app:models}

Since our results concern different computational models, 
we use this appendix to describe them in more detail. 
Our two models are the real RAM/pointer machine and 
the word RAM.
\paragraph{The Real RAM/Pointer Machine.}
The standard model in computational geometry
is the \emph{real RAM}. Here, data is represented
as an infinite sequence of storage cells. These cells
can be of two different types: they can store real numbers
or integers. The model supports standard operations on
these numbers in constant time, including addition, 
multiplication, and elementary functions
like square-root.
The \emph{floor} function can be used to truncate a 
real number to an integer, but if we were allowed to 
use it arbitrarily, the real RAM could solve PSPACE-complete
problems in polynomial time~\cite{Schoenhage79}.
Therefore, we usually have only a restricted floor function 
at our disposal,
and in this paper it will be banned altogether. 

The  \emph{pointer machine}~\cite{Knuth97} models 
the list processing capabilities
of a computer and
disallows the use of constant time table lookup.
The data structure is modeled as a directed
graph $G$ with bounded out-degree. Each node in
$G$ represents a \emph{record}, with a bounded
number of pointers to other records and a bounded number
of (real or integer) data items.  
The algorithm can access data only by following pointers
from the inputs  (and a bounded number of global entry
records); random access is not possible. The data can be 
manipulated through the usual real RAM operations (again, 
we disallow the floor function).

\paragraph{Word RAM.}
The \emph{word RAM} is essentially a real RAM without
support for real numbers. However, on a real RAM,
the integers are usually treated as atomic, whereas
the word RAM allows for powerful bit-manipulation tricks.
More precisely, the word RAM represents the 
data as a sequence of
$w$-bit words, where $w \geq \log n$ ($n$ being the
problem size).
Data can be accessed arbitrarily, and standard
operations, such as Boolean operations 
(\texttt{and}, \texttt{xor}, \texttt{shl}, $\ldots$), addition, or
multiplication take constant time. There are many variants of
the word RAM, depending on precisely which instructions are 
supported in constant time. The general consensus seems
to be that any function in $\text{AC}^0$
is acceptable.\footnote{$\text{AC}^0$ is the 
class of all functions $f: \{0,1\}^* \rightarrow \{0,1\}^*$ that
can be computed by a family of circuits $(C_n)_{n \in \mathbb{N}}$ with the 
following properties: (i) each $C_n$ has $n$ inputs; (ii) there exist constants
$a,b$, such that $C_n$ has at most $an^b$ gates, for $n\in \mathbb{N}$; 
(iii) there is a constant $d$ such that for all $n$ the length of the longest
path from an input to an output in $C_n$ is at most $d$ (i.e., the
circuit family has bounded depth); (iv) each gate
has an arbitrary number of incoming edges (i.e., the \emph{fan-in} is 
unbounded).} However, it is always preferable to rely on a set of operations
as small, and as non-exotic, as possible.
Note that multiplication is not in $\text{AC}^0$~\cite{FurstSaSi84}. 
Nevertheless, it is usually
included in the word RAM instruction set~\cite{FredmanWi94}.

\end{document}